\documentclass{article}

\IfFileExists{arxiv.sty}{
  \usepackage{arxiv}
}{
  \usepackage[margin=1in]{geometry}
  
  \providecommand{\keywords}[1]{}
}

\usepackage[utf8]{inputenc} % allow utf-8 input
\usepackage[T1]{fontenc}  % use 8-bit T1 fonts
\usepackage{url}      % simple URL typesetting
\usepackage{booktabs}    % professional-quality tables
\usepackage{amsfonts}    % blackboard math symbols
\usepackage{nicefrac}    % compact symbols for 1/2, etc.
\usepackage{microtype}   % microtypography
\usepackage{graphicx}

\usepackage{amsmath,amsthm,amsfonts,amscd} 
\usepackage{eucal} 	 	% Euler fonts
\usepackage{verbatim}   	% Allows quoting source with commands.
\usepackage{makeidx}    	% Package to make an index.
\usepackage{graphicx}      % Including graphics
\usepackage{cite}     	% 
\usepackage{url}		% Allows good typesetting of web URLs.
\usepackage[font=footnotesize,labelfont=bf,center]{caption}
\usepackage{subcaption}
\usepackage{float}
\usepackage{algpseudocode}
\usepackage{listings} 
\IfFileExists{listings-rust.sty}{
  \usepackage{listings-rust}
}{}
\definecolor{rustkey}{RGB}{170, 55, 0}        % burnt orange for keywords
\definecolor{rusttype}{RGB}{44, 120, 120}      % teal for types/traits
\definecolor{ruststr}{RGB}{58, 135, 58}        % muted green for strings
\definecolor{rustcmt}{RGB}{120, 120, 120}      % gray for comments
\definecolor{rustmacro}{RGB}{130, 80, 160}     % purple for macros/attributes

\lstdefinelanguage{Rust}{
  keywords={
    fn, pub, struct, trait, impl, where, self, mut,
    let, return, if, else, for, while, loop, break,
    continue, match, type, enum, const, static, unsafe,
    use, mod, crate, super, ref, move, async, await, true, false
  },
  keywords=[2]{
    bool, usize, u32, u64, i32, i64, f32, f64, str, String,
    Vec, Option, Result, Some, None, Ok, Err, Box, Arc,
    Send, Sync, Atomize, AtomicBool, AtomicU32,
    GlobalState, LatticeLinearProblem, Worklist, NullWorklist,
    Ordering, UnsafeCell, Debug
  },
  sensitive=true,
  comment=[l]{//},
  morecomment=[s]{/*}{*/},
  string=[b]",
  morestring=[b]',
}
\usepackage{titlesec}
\titlespacing*{\section}{0pt}{12pt plus 2pt minus 2pt}{6pt plus 1pt minus 1pt}
\titlespacing*{\subsection}{0pt}{10pt plus 2pt minus 2pt}{4pt plus 1pt minus 1pt}
\titlespacing*{\subsubsection}{0pt}{8pt plus 2pt minus 2pt}{4pt plus 1pt minus 1pt}
\usepackage[linesnumbered,ruled,vlined]{algorithm2e}
\SetAlFnt{\small}
\SetAlCapFnt{\small}
\SetAlCapNameFnt{\small}
\usepackage{xcolor}
\usepackage[hidelinks]{hyperref}
\usepackage{tikz}
\usetikzlibrary{positioning, arrows, arrows.meta}

\newtheorem{lemma}{Lemma}

\newtheorem{definition}{Definition}
\newtheorem{claim}{Claim}

\newcommand{\CC}{{L}}

\newcommand{\C}{G}
\newcommand{\RR}{\mathbf{R}}

\DeclareMathOperator{\forbidden}{forbidden}

\title{A common parallel framework for LLP combinatorial problems}

\author{
 David Ribeiro Alves \\
 Department of Electrical and Computer Engineering\\
 University of Texas at Austin\\
 Austin, TX 78712 \\
 \texttt{dralves@utexas.edu} \\
  \and
 Vijay K. Garg \\
 Department of Electrical and Computer Engineering\\
 University of Texas at Austin\\
 Austin, TX 78712 \\
 \texttt{garg@utexas.edu} \\
}

\begin{document}
\maketitle
\begin{abstract}
    Traditional lock-free parallel algorithms for combinatorial optimization problems, such as shortest paths, stable matching, and job scheduling require programmers to write problem-specific routines and synchronization code. We propose a general-purpose lock-free runtime, LLP-FW that can solve all combinatorial optimization problems that can be formulated as a Lattice-Linear Predicate by advancing all forbidden local states in parallel until a solution emerges. The only problem-specific code is a definition of the forbiddenness check and a definition of the advancement. We show that LLP-FW can solve several different combinatorial optimization problems, such as Single Source Shortest Paths (SSSP), Breadth-First Search (BFS), Stable Marriage, Job Scheduling, Transitive Closure, Parallel Reduction, and 0-1 Knapsack. We compare LLP-FW against hand-tuned, custom solutions for these seven problems and show that it compares favorably in the majority of cases.
\end{abstract}

% keywords can be removed

%%%%%%%%%%%%%%%%%%%%%%%%%%%%%%%%%%%%%%%%%%%%%%%%%%%%%%%%%%%%%%%%%%%%%%
% Introduction for the Arxiv version
%%%%%%%%%%%%%%%%%%%%%%%%%%%%%%%%%%%%%%%%%%%%%%%%%%%%%%%%%%%%%%%%%%%%%%

\section{Introduction}
\label{sec:intro}

Combinatorial optimization problems like \emph{Single Source Shortest Path} (SSSP), \emph{Stable Marriage} (SM), and constrained \emph{Job Scheduling} problems are common in transportation networks, matching markets, and scheduling platforms. Each of these problems has well-established classical algorithms. However, achieving efficient parallel solutions for these problems is usually problem-specific and requires unique synchronization techniques, data structures, and tuning strategies. This requires duplicate engineering effort for each problem and raises the potential for lost opportunities of cross-optimization between problems.

\emph{Lattice-Linear Predicates} (LLP) are a theoretical abstraction that can be used to model a wide range of combinatorial optimization problems. LLPs specify a monotone property over the lattice of feasible solutions such that the violation of the property in the entire solution space implies the presence of at least one \emph{forbidden} local solution. The advancement of forbidden solutions in the solution space leads to the monotone search property, which converges to the feasible (and in many cases optimal) solution. LLPs were first introduced for the detection of predicates in distributed systems~\cite{Chase1998,Garg2015} and later established their power in representing optimization problems~\cite{Garg2018}. In prior work, we applied the LLP abstraction to derive problem-specific parallel algorithms for two classical graph problems: Single Source Shortest Path~\cite{Alves.2020}, and  Minimum Spanning Tree~\cite{Alves2022MST}. Both sets of algorithms achieved competitive or superior performance relative to established baselines, demonstrating that the LLP formulation can produce practical parallel algorithms. However, each of these implementations was hand-written for a specific problem, with its own synchronization code, data structures, and tuning. This paper asks whether the pattern common to both---detecting forbidden states and advancing them in parallel---can be extracted into a reusable runtime that works across many problems without reimplementing concurrency control from scratch.

In this paper we present a single, lock-free, framework, \emph{LLP-FW} (Lattice Linear Predicate Framework), that can be applied across diverse problems without re-implementing concurrency primitives from scratch. We demonstrate this approach on large-scale weighted and unweighted graph problems (SSSP, BFS, Transitive Closure), preference-based matching (Stable Marriage), and precedence-constrained scheduling and optimization kernels (Job Scheduling, Reduction, Knapsack). Throughout all of these domains, progress is uniformly described as the advancement of local forbidden states, and problem-specific logic is encapsulated as lightweight adapters. A key advantage of this design is that, as long as a problem can be expressed as an LLP, the problem specification and the solvers are fully decoupled. This means that a practitioner can implement a new problem by defining only the forbidden-state predicate and the advance function, and then run all available solvers on it to find the best-performing one without writing any concurrency code. Conversely, if a new solver strategy or a new optimization technique within an existing solver is developed, it automatically applies to every problem in the framework without changing any problem-specific code.

Our evaluation shows that LLP-FW achieves strong gains on problems with narrow forbidden frontiers and cascading dependencies: up to \(246\times\) over the parallel Gale--Shapley baseline on Stable Marriage, \(23\times\) over Floyd--Warshall on sparse transitive closure, and \(16\times\) on road-network BFS at 32 threads. On problems where the baseline scales poorly (e.g., knapsack), LLP-FW achieves \(1.8\)--\(3.8\times\) gains. The framework does lose on regular, bandwidth-limited workloads like parallel reduction (\(4\)--\(5\times\) slower), where the per-operation atomic overhead dominates. Across all seven problems, we find that the shape of the forbidden frontier---narrow vs.\ wide---is the best predictor of whether LLP-FW will outperform a specialized baseline.

Our contributions are fourfold:
\begin{enumerate}
 \item \emph{Generic lock-free LLP runtime.} We propose a generic shared memory runtime for combinatorial optimization, where generic parallel scheduling is decoupled from problem-specific forbidden and advance definitions.
 \item \emph{Systems implementation.} Our runtime is implemented using atomic operations and lock-free worklists, both of shared and per-thread types.
 \item \emph{Instantiations for several combinatorial optimization problems.} We apply our generic runtime to several combinatorial optimization problems, including SSSP, BFS, stable marriage, job scheduling, transitive closure, reduction, and knapsack.
 \item \emph{Unified evaluation of our approach.} Our approach is evaluated against strong baselines under a uniform evaluation methodology, highlighting where LLP-FW compares favorably or unfavorably to problem-specific solutions and where gaps remain for future research.
\end{enumerate}

The rest of this paper is organized as follows. Section~\ref{sec:cm-rw} surveys related work in parallel frameworks, lock-free shared memory algorithms, and optimization with LLP. Section~\ref{sec:cm-llp-bg} describes the theoretical background of lattice-linear predicates. Section~\ref{sec:cm-llp-fw} describes the design and implementation of our approach. Section~\ref{sec:cm-eval} presents our evaluation results and discusses our findings. Section~\ref{sec:cm-conclusion} concludes our paper with some discussions on future research directions.

\section{Related Work}
\label{sec:cm-rw}

\subsection{Lattice-Linear Predicates in Optimization}

The concept of lattice-linearity is first introduced as a result of global predicate detection in distributed systems~\cite{Chase1998,Garg2015}. A predicate is said to be lattice-linear if all the violating states contain at least one forbidden coordinate, which is also locally advanceable. This enables the monotonic correction process, ensuring convergence to a feasible state. Subsequently, Garg~\cite{Garg2018} showed that this abstraction is capable of representing a wide variety of combinatorial optimization problems, such as shortest path, stable marriage, and market prices, under a single mathematical structure, with classical algorithms as special cases.

The LLP approach has been further extended in several other directions, with dynamic programming problems, such as the knapsack problem, being modeled and solved using parallel algorithms similar to the LLP approach~\cite{Garg2022ICDCN}. The problem of stable marriage has also been related to predicate detection approaches~\cite{Garg2017DISC}. $SP_1$, $SP_2$, and $ParSP_2$ shortest path algorithms are also derived using the LLP approach, with the concept of fixedness being used to compute the distance value earlier than classical approaches~\cite{Alves.2020,Garg.2020}. Similarly, the minimum spanning tree problem is also modeled using the LLP approach for predicate detection~\cite{Alves2022MST}. Very recently, the concept of equilevel predicate detection is introduced, which is an extension of the lattice predicate approach with new results on efficient online parallel detection~\cite{Garg2024Equilevel}.

LLP-FW is an extension of the above approach, with the concept of LLP being implemented using a shared memory runtime environment, with the rules of forbiddenness and advancement being implemented in parallel, with the runtime environment providing the necessary tools for updating the states atomically.

\subsection{Specific Problem Formulations}

There are a number of problems studied within this work for which there is an existing algorithmic history, and this history is relevant to the LLP formulation.

The Single Source Shortest Path (SSSP) problem is arguably one of the most studied combinatorial optimization problems. The traditional sequential algorithms for this problem include Dijkstra's algorithm~\cite{Dijkstra1959}, the Bellman-Ford algorithm~\cite{Bellman1958}, and Johnson's algorithm~\cite{Johnson1977}. The parallel version of the SSSP problem has also been extensively studied; for instance, the Delta Stepping method~\cite{Meyer2003} partitions the graph into “distance buckets” to achieve parallel computation, although the parameters need to be set appropriately to balance sparsity and density. The LLP formulation of the SSSP problem~\cite{Garg2018} uses the coordinates of the lattice corresponding to the distances between vertices and the forbidden state advancement mechanism to facilitate the solution of the relaxation operation.

The Stable Marriage (SM) problem is traditionally solved using the Gale-Shapley algorithm, where men propose to women based on their preferences~\cite{Gale1962,Gusfield1989,Knuth1997}. The LLP formulation of the Stable Marriage Problem~\cite{Garg2018,Garg2017DISC} uses the preference index of each man as the coordinates of the solution vector in the lattice, effectively providing a parallel solution to the stable marriage problem while ensuring correctness within the framework of the lattice-linear paradigm.

The problem of scheduling jobs with prerequisites is equivalent to computing critical-path lengths in a Directed Acyclic Graph~\cite{Kelley1959}. This is naturally equivalent to the LLP formulation, where the forbidden state advancement mechanism is used to enhance the solution feasibility in a monotonic manner, similar to the graph and the stable marriage problems.

\subsection{Parallel Frameworks and Separation of Algorithm and Scheduler}

Substantial research has been conducted on abstractions for separating the expression of an algorithm from the expression of the execution schedule for graph and graph-like computations. Distributed vertex-centric frameworks, such as Pregel~\cite{Malewicz2010}, partition the computation into supersteps of message passing and have been successful for graph analytics applications. On the other hand, the GraphLab~\cite{Low2010} approach focuses on asynchronous graph-parallel computation with explicit consistency models for sparse dependency structures. The Galois and Ligra~\cite{Pingali2011, Shun2013} approaches emphasize efficient graph processing using conflict-aware scheduling and frontier-based approaches, respectively. More recent research, such as the GraphIt approach~\cite{Zhang2018}, has demonstrated the value of separating algorithm expression from a scheduling language for exploring new approaches to locality, parallelism, and traversal for efficient graph analytics.

The key difference is that the approach of LLP-FW is predicate-driven, focusing on forbidden states and advancement rules, and is lock-free, focusing on shared memory and moving beyond graph-based computations, such as the stable marriage and knapsack problems, rather than committing to a particular vertex-centric graph computation approach.

\subsection{Lock-Free Algorithm and Scheduling}

Lock-free and wait-free algorithms~\cite{Herlihy1991,Herlihy2008} leverage atomicity to achieve progress properties under high contention. In the context of task parallelism, one of the most popular scheduling methodologies for efficiently managing irregular parallelism is the use of work stealing. Some of the foundational results in this space have established the efficacy of the scheduling approach for parallel environments~\cite{Blumofe1999}. Parallel graph algorithms heavily rely on lock-free approaches for updating distance, frontier, and visitation. However, they also involve phased barriers and synchronization during the execution process.

The proposed work extends the existing systems insights, including private work, work stealing for parallel graph algorithms, and investigates problem-aware scheduling methodologies, including recency bias for graph relaxation, chunked FIFO for cascading matching, and bucket scheduling for structured dynamic programming state spaces. Moreover, the proposed work also demonstrates that the optimal scheduling approach depends on the frontier shapes, along with their locality properties.

\section{Background}

\label{sec:cm-llp-bg}

In this section we introduce the theoretical background to understand the \emph{LLP-FW} framework. For a more in-depth discussion of the theoretical underpinnings on these topics, readers are directed to the works of Chase and Garg~\cite{chase1998detection} and Garg~\cite{Garg.2020}.

\paragraph{Lattices and Global States.}

We consider a lattice denoted by \(\CC\) defined by all \(n\)-dimensional vectors of nonnegative real numbers, with an upper bound defined by the vector \(T\). Every element of this lattice, denoted by \(\C \in \CC\), is an \(n\)-dimensional vector defined by:

\[
\C = ( \C[1], \C[2], \ldots, \C[n] ),
\]
such that \( 0 \leq \C[i] \leq T[i] \) for all \(i\). The partial order on this lattice is defined component-wise, i.e.,:
\[
\C \leq \C' \quad \text{iff} \quad \C[i] \leq \C'[i] \quad \forall i.
\]

This model is appropriate to describe a notion of a \emph{search space} of solutions, applicable to a variety of combinatorial optimization problems. We consider a scenario with \(n\) processes, where each process \(i\) has a single dimension denoted by \(\C[i]\). We call this a \emph{global state} denoted by \(\C\), while the local state of process \(i\) is denoted by \(\C[i]\).

\paragraph{Example Poset and Lattice.}

Figure~\ref{fig:poset-lattice} illustrates a simple finite partial order (poset) scenario when \(n = 2\). The associated distributive lattice comprises all combinations of local states that satisfy this partial order. In this particular scenario, the number of global states is eleven.

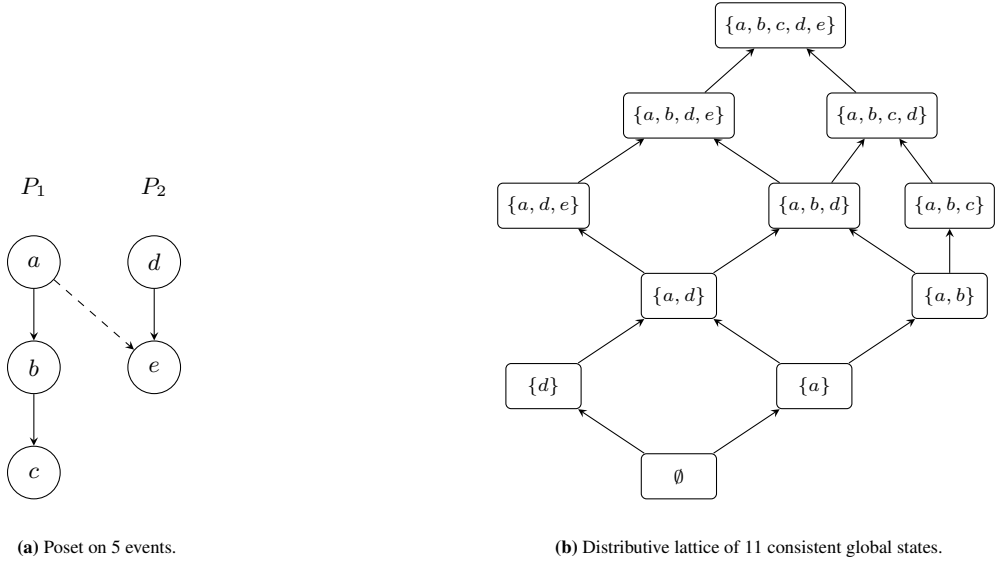
\begin{figure}[htbp]
\centering
\begin{minipage}[b]{0.3\textwidth}
\centering
\begin{tikzpicture}[>=stealth,
  evt/.style={circle, draw, minimum size=0.7cm, font=\small}]
  % Process 1 (left column)
  \node at (-0.8,4.2) {\footnotesize $P_1$};
  \node[evt] (a) at (-0.8,3.2) {$a$};
  \node[evt] (b) at (-0.8,1.8) {$b$};
  \node[evt] (c) at (-0.8,0.4) {$c$};
  \draw[->] (a) -- (b);
  \draw[->] (b) -- (c);
  % Process 2 (right column)
  \node at (0.8,4.2) {\footnotesize $P_2$};
  \node[evt] (d) at (0.8,3.2) {$d$};
  \node[evt] (e) at (0.8,1.8) {$e$};
  \draw[->] (d) -- (e);
  % Cross-dependency
  \draw[->, dashed] (a) -- (e);
\end{tikzpicture}
\vspace{0.3cm}
\subcaption{Poset on 5 events.}
\end{minipage}%
\hfill
\begin{minipage}[b]{0.65\textwidth}
\centering
\begin{tikzpicture}[>=stealth,
  box/.style={draw, rectangle, rounded corners=2pt, minimum width=1.0cm, minimum height=0.6cm, font=\scriptsize, align=center}]
  \node[box] (00) at (0,0)      {$\emptyset$};
  \node[box] (01) at (-1.8,1.2) {$\{d\}$};
  \node[box] (10) at (1.8,1.2)  {$\{a\}$};
  \node[box] (11) at (0,2.4)    {$\{a,d\}$};
  \node[box] (20) at (3.6,2.4)  {$\{a,b\}$};
  \node[box] (12) at (-1.8,3.6) {$\{a,d,e\}$};
  \node[box] (21) at (1.8,3.6)  {$\{a,b,d\}$};
  \node[box] (30) at (3.6,3.6)  {$\{a,b,c\}$};
  \node[box] (22) at (0,4.8)    {$\{a,b,d,e\}$};
  \node[box] (31) at (2.7,4.8)  {$\{a,b,c,d\}$};
  \node[box] (32) at (1.35,6.0) {$\{a,b,c,d,e\}$};
  % Edges (Hasse diagram)
  \draw[->] (00) -- (01);
  \draw[->] (00) -- (10);
  \draw[->] (01) -- (11);
  \draw[->] (10) -- (11);
  \draw[->] (10) -- (20);
  \draw[->] (11) -- (12);
  \draw[->] (11) -- (21);
  \draw[->] (20) -- (21);
  \draw[->] (20) -- (30);
  \draw[->] (12) -- (22);
  \draw[->] (21) -- (22);
  \draw[->] (21) -- (31);
  \draw[->] (30) -- (31);
  \draw[->] (22) -- (32);
  \draw[->] (31) -- (32);
\end{tikzpicture}
\vspace{0.3cm}
\subcaption{Distributive lattice of 11 consistent global states.}
\end{minipage}
\caption{(a)~A computation poset for $n=2$ processes. $P_1$ has events $a, b, c$ and $P_2$ has events $d, e$. The dashed arrow indicates a cross-process dependency: $a$ must complete before $e$. (b)~The corresponding distributive lattice. Each node shows the set of completed events. The state $\{d,e\}$ is absent because $e$ depends on $a$. Arrows indicate single-event transitions.}
\label{fig:poset-lattice}
\end{figure}

% Consolidate SSSP content under the SSSP section above
% (hide duplicate narrative/table/figure further down)
\paragraph{Combinatorial Optimization and Predicate Detection.}
Given this lattice, a combinatorial optimization problem amounts to finding the \emph{minimum} element of \(\CC\) that satisfies a given Boolean predicate \(B\). For instance, in a shortest-path problem, \(\CC\) would capture all candidate distance vectors and \(B\) would encode whether these distances are correct. The goal is then:
\[
\min_{\C \in \CC} \bigl\{ \C \,\big\vert\, B(\C) \text{ is true} \bigr\}.
\]
Finding whether there exists a \(\C \in \CC\) that satisfies \(B\) is known as the \emph{predicate detection} problem. In the problems we examine here, we want not just any feasible \(\C\), but the minimal one.

\paragraph{Forbidden States and Lattice-Linearity.}
The central concept in this framework is the notion of a \emph{forbidden} state, introduced by Chase and Garg~\cite{chase1998detection}. A local state $G[j]$ is forbidden if keeping $G[j]$ at its current value means that the predicate $B$ can never become true, regardless of what happens to the other coordinates. Formally:

\begin{definition}[Forbidden State~\cite{chase1998detection}]\label{def:forbidden}
 Let $\CC$ be a distributive lattice of $n$-dimensional vectors of $\RR_{\ge 0}$, and $B$ a boolean predicate on $\CC$. For any global state $G \in \CC$, the index $j$ (or the local state $G[j]$) is \emph{forbidden} if:
 \[
  \forbidden(G,j,B) \;\equiv\; \forall H \in \CC : G \leq H : (G[j] = H[j]) \Rightarrow \neg B(H).
 \]
\end{definition}

\noindent
Thus, if $G$ violates $B$, and $j$ is forbidden in $G$, we know we must \emph{advance} $G[j]$ in order to approach a feasible solution. A predicate $B$ is called \emph{lattice-linear} if any global state not satisfying $B$ must contain at least one forbidden index:

\begin{definition}[Lattice-Linear Predicate~\cite{chase1998detection}]
 A boolean predicate $B$ is \emph{lattice-linear} with respect to $\CC$ if:
 \[
  \forall G \in \CC: \neg B(G) \;\Rightarrow\; \exists j : \forbidden(G,j,B).
 \]
\end{definition}

The important thing to note is that whenever $B(G)$ is false, we can always find at least one index $j$ to advance locally, which is what makes this approach practical for parallel computation.

\paragraph{Advancing Forbidden States.}
In practice, we also need to specify \emph{how much} to advance a forbidden coordinate. This leads to a refinement of the definition of forbiddenness:

\begin{definition}[$\alpha$-forbidden]
 Let $B$ be a boolean predicate on $\CC$. A coordinate $j$ of $\C$ is \emph{$\alpha$-forbidden} if
 \[
  \forall H \in \CC : H \geq \C : (H[j] < \alpha) \;\Rightarrow\; \neg B(H),
 \]
 where $\alpha > \C[j]$ is a problem-dependent increment usually encoded in the `"advance"` function''. 
\end{definition}

What this means in practice is that if $\C[j]$ is $\alpha$-forbidden, then any solution where $\C[j]$ stays below $\alpha$ is guaranteed to violate $B$, so we must increase $\C[j]$ at least to $\alpha$.

\paragraph{Algorithm \textsc{LLP}.}
Putting this all together, Algorithm~\ref{fig:alg-llp} (adapted from~\cite{chase1998detection,Garg.2020}) shows how to find the least vector at most $T$ that satisfies a lattice-linear predicate $B$:

\begin{algorithm}[H]
\SetAlgoRefName{LLP}
\caption{To find the minimum vector at most $T$ that satisfies $B$.}
\label{fig:alg-llp}
\DontPrintSemicolon
\KwIn{$T$: vector, $B$: lattice-linear predicate}
\KwOut{vector $\C$, or \texttt{null} if no feasible solution exists}

\Begin{
 Initialize $\C$ such that $\C[i] = 0$ for all $i$.\;
 \While{ $\exists j: \forbidden(\C,j,B)$ }{
  \ForEach{$j$ s.t. $\forbidden(\C,j,B)$}{
   \If{$\alpha(\C,j,B) > T[j]$}{
    \Return{\texttt{null}}
   }
   \Else{
    $\C[j] \gets \alpha(\C,j,B)$
   }
  }
 }
 \Return{$\C$}\;
}
\end{algorithm}

\noindent

The algorithm works by finding all the coordinates $j$ that are currently forbidden and moving them forward, simultaneously. It keeps doing this until no coordinate is forbidden, i.e., until $B(\C)$ is satisfied. If at any point the required advance would exceed the bound $T[j]$, then no solution exists within the given bounds.

The reason the algorithm terminates is that each coordinate only moves forward (monotonically) and there are finitely many possible values. Since $B$ is lattice-linear, whenever $B(G)$ is false there is at least one coordinate to advance, so the algorithm makes progress in every iteration. Lemma~\ref{lem:basic-LLP}, adapted from Garg~\cite{Garg.2020} and Chase and Garg~\cite{chase1998detection}, shows when multiple lattice-linear predicates can be combined.

\begin{lemma}[Basic Lattice-Linearity~\cite{Garg.2020,chase1998detection}]
\label{lem:basic-LLP}
Let $B$ be any boolean predicate on a lattice $\CC$ of vectors. 
\begin{enumerate}
\item[(a)] Suppose $f \colon \CC \to \RR_{\ge 0}$ is a monotone function, and consider a predicate $B \equiv \{\C \mid \C[i] \ge f(\C)\}$ for a fixed $i$. Then $B$ is lattice-linear.
\item[(b)] If $B_1$ and $B_2$ are each lattice-linear, then $B_1 \wedge B_2$ is also lattice-linear.
\end{enumerate}
\end{lemma}

\noindent
This is important because it means that conditions from different aspects of a problem (e.g., distance constraints in shortest paths, matching constraints in stable marriage) can be combined into a single lattice-linear predicate that the \textsc{LLP} algorithm can then solve.

While the theory above gives us the foundations, there are several practical challenges that are not directly addressed in the original theoretical presentation. In our work we focus on:
\begin{itemize}
  \item selecting which states to test for forbiddenness in very large lattices;
  \item predicting which states become newly forbidden after an advance;
  \item prioritizing candidate updates that prune the search space most effectively; and
  \item ensuring that concurrent updates preserve lattice monotonicity and convergence.
\end{itemize}

The next section describes how we address these challenges in \emph{LLP-FW}.

\section{LLP framework}

\label{sec:cm-llp-fw}

In this section we present LLP-FW, our shared memory runtime for solving combinatorial problems that can be expressed as lattice linear predicates. The runtime is built around three operations: initialization, forbidden checks, and advancements. These operations are executed in parallel using lock-free solvers that iterate until a solution is found. We describe the implementation model, discuss the concurrency and scheduling issues that arise, and compare different solver variants in terms of their complexity and performance characteristics.

\subsection{Illustrative Example}

\label{sec:cm-llp-example}

To exemplify the benefits and difficulties of using the LLP formulation, consider the simple weighted graph in Figure~\ref{fig:example-graph}. The labels of the edges represent the different relaxations in the lattice transitions.

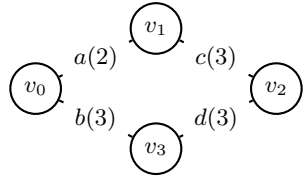
\begin{figure}[h!]
  \centering
  \begin{tikzpicture}[scale=0.8]
    % Nodes
    \begin{scope}[every node/.style={circle, thick, draw, font=\small}]
      \node (A) at (0,2) {$v_0$};
      \node (B) at (2,1) {$v_3$};
      \node (C) at (4,2) {$v_2$};
      \node (D) at (2,3) {$v_1$};
    \end{scope}
    % Edges
    \begin{scope}[>={Stealth[black]},
           every node/.style={fill=white, circle, font=\small},
           every edge/.style={draw=black, thick}]
      \path [-] (A) edge node {$b(3)$} (B);
      \path [-] (A) edge node {$a(2)$} (D);
      \path [-] (B) edge node {$d(3)$} (C);
      \path [-] (C) edge node {$c(3)$} (D);
    \end{scope}
  \end{tikzpicture}
  \caption{An undirected graph with 4 nodes.}
  \label{fig:example-graph}
\end{figure}

\smallskip

In this example, we solve Single Source Shortest Path (SSSP) from node 0. Each node \(i\) stores a distance estimate \(\C[i]\). The initial state is \(\C = [0, \infty, \infty, \infty]\), and each transition corresponds to an edge relaxation. Figure~\ref{fig:example-lattice} shows the resulting lattice of global states.

\begin{figure}[htb]
	\centering
	\begin{tikzpicture}[scale=0.58]
		\begin{scope} [
				auto,
				every node/.style={minimum size=0.9cm, text centered, font=\tiny},
				>=latex
			]
			\node (C0) at (0,0) {\( \C_0 = \begin{bmatrix} 0 \\ \infty \\ \infty \\ \infty \end{bmatrix} \)};
 
			\node (C1) at (-2,3) {\( \C_1 = \begin{bmatrix} 0 \\ 2 \\ \infty \\ \infty \end{bmatrix} \)};
			\node (C2) at (2,3) {\( \C_2 = \begin{bmatrix} 0 \\ \infty \\ \infty \\ 3 \end{bmatrix} \)};
 
			\node (C3) at (-4,6) {\( \C_3 = \begin{bmatrix} 0 \\ 2 \\ 5 \\ \infty \end{bmatrix} \)}; 
			\node (C4) at (-0.5,6) {\( \C_4 = \begin{bmatrix} 0 \\ 2 \\ \infty \\ 3 \end{bmatrix} \)};
			\node (C6) at (4,6) {\( \C_6 = \begin{bmatrix} 0 \\ \infty \\ 6 \\ 3 \end{bmatrix} \)};
 
			\node (C7) at (-6,9) {\( \C_7 = \begin{bmatrix} 0 \\ 2 \\ 5 \\ 8 \end{bmatrix} \)};
			\node (C8) at (6,9) {\( \C_8 = \begin{bmatrix} 0 \\ 9 \\ 6 \\ 3 \end{bmatrix} \)};

      \node (C9) at (2,12) {\( \C_9 = \begin{bmatrix} 0 \\ 2 \\ 6 \\ 3 \end{bmatrix} \)};
 
			\node (CF) at (0,15) {\( \C_F = \begin{bmatrix} 0 \\ 2 \\ 5 \\ 3 \end{bmatrix} \)};
		\end{scope}
 
		\begin{scope}[>={Stealth[black]},
				every node/.style={fill=white, font=\tiny},
			  every edge/.style={draw, ->}]
			\path (C0) edge node {$a$} (C1);
			\path (C0) edge node {$b$} (C2);
 
			\path (C1) edge node {$c$} (C3);
			\path (C1) edge node {$b$} (C4);
 
			\path (C2) edge node {$a$} (C4);
			\path (C2) edge node {$d$} (C6);
 
			\path (C3) edge node {$b$} (CF);       
			\path (C3) edge node {$d$} (C7);

        \path (C4) edge node {$c$} (CF);
        \path (C4) edge node {$d$} (C9);

        \path (C6) edge node {$a$} (C9);
        \path (C6) edge node {$c$} (C8);
 
			\path (C7) edge node {$b$} (CF);
			\path (C9) edge node {$c$} (CF);
        \path (C8) edge node {$a$} (C9);
		\end{scope}
	\end{tikzpicture}
	\caption{Lattice of global states with all state transitions through all possible edge relaxations for the example graph.}
  \label{fig:example-lattice}
\end{figure}
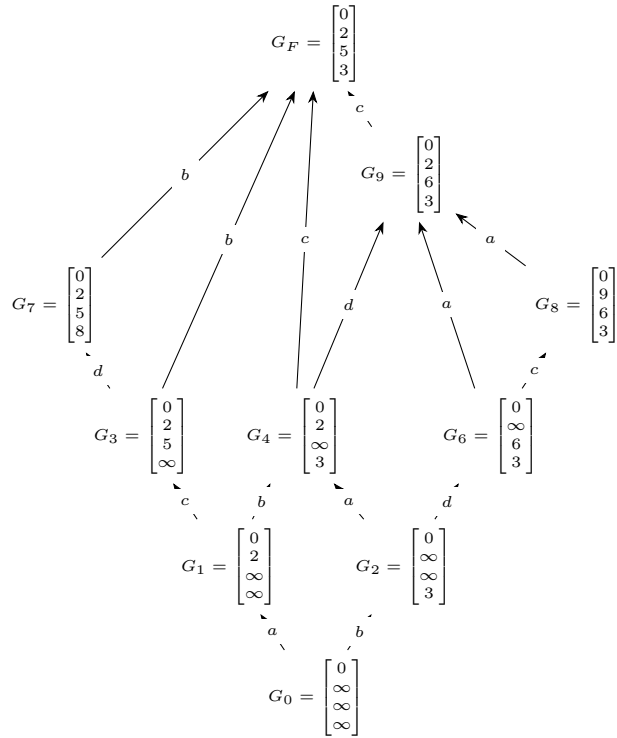

This small example already has a number of valid relaxation paths leading to the same final state. This number grows extremely fast with the size of the graph, so it is impossible to check all of these paths. The length of the paths is also important, e.g., \( \C_0 \to \C_1 \to \C_3 \to \C_F \) has three steps, while \( \C_0 \to \C_2 \to \C_6 \to \C_8 \to \C_9 \to \C_F \) has five. This means more work has to be done, even if correctness is preserved.

However, it also shows that parallelism is possible. At state \( \C_0 \), \( \C[1] \) and \( \C[3] \) are forbidden, so they can be relaxed simultaneously. This means that regardless of which finishes first, the resulting state is within the lattice, so this does not affect correctness. This duality of correctness and schedule quality is a key factor that makes it possible to apply the LLP-FW parallelization to different problem domains.

\subsubsection{Splitting Work Between Threads} In a way, splitting work between threads in the LLP is a problem that, despite being quite different, has a number of similarities with splitting work in any other parallel processing system. One naive way to solve this would be to give each thread exclusive write access to a subset of elements, while at the same time giving it read access to all elements it might need to update. This would require the use of locks, etc., to ensure correct program behavior.

Nevertheless, such a strong dependence on locks may result in a significant degradation of the program’s performance. Locks impose a serial constraint on accessing the shared resource, and threads will be forced to wait for each other, thus defeating the purpose of using parallelism. The contention among threads to access the lock may result in a bottleneck, latency, and underutilization of CPU cores.

To avoid such problems, we will use an optimistic approach to handling concurrency issues. Here, the threads will proceed with updating the shared variables without acquiring any lock on those variables. This approach accepts the possibility of having many threads performing redundant updates on the same element, as the cost of such redundant updates may be less compared to the benefits of using an optimistic approach.

In our problem, this means that the threads will proceed with updating the elements of the solution vector independently. Atomic variables will be used to ensure the safety of updating the shared variables without acquiring any lock on those variables. If a thread is unable to update the element of the solution vector due to the updates performed by other threads, it will simply retry the operation.

Furthermore, certain problems benefit from prioritizing the processing of specific elements. For instance, in the Single Source Shortest Path (SSSP) problem, nodes with lower distance estimates are more likely to become fixed—that is, reach their final shortest path distances—and require no further updates. By processing these nodes first, we can reduce unnecessary work on nodes that might otherwise be updated multiple times. This prioritization aligns with the principle of exploiting problem-specific knowledge to optimize parallel processing.

Implementing a priority mechanism, however, introduces additional complexity. It may necessitate priority queues or other data structures that can handle concurrent access while maintaining element ordering based on priority and that can be in itself a bottleneck. Care must be taken to balance the benefits of priority-based processing with the overhead of maintaining such data structures in a parallel environment.

In summary, effective processing and partitioning of work among threads in the LLP framework involves:

\begin{enumerate}
  \item \textbf{Only looking at forbidden states that are likely to change:} For problems of meaningful size, looking at all the solution vector elements at each iteration is infeasible since the there are much less threads than there are elements.
  \item \textbf{Optimistic Concurrency:} Allowing threads to handle conflicts through atomic operations and retries.
  \item \textbf{Reducing cross-thread contention and communication:} Keeping data thread-local when possible and avoiding locks where possible to prevent serialization and performance degradation. 
  \item \textbf{Prioritized Processing:} Leveraging problem-specific characteristics to process high-priority elements first, thereby reducing redundant computations.
\end{enumerate}

Section~\ref{sec:cm-llp-practice} presents solver variants that address these trade-offs in different ways and perform well on different workloads. We first examine optimistic concurrent updates in more detail and highlight the main correctness pitfalls.

\subsection{Efficient State Selection in Practice}

While LLP theory describes how to identify and advance forbidden states, it does not prescribe how to select candidate states efficiently. At practical scales, this selection policy can dominate runtime. The state selection strategies discussed below are well-established in worklist-based algorithms (e.g., priority-driven relaxation in Dijkstra-style solvers, neighbour-push in BFS); we describe them here in terms of the LLP abstraction to clarify how they map onto the framework's solver variants.

\paragraph{The State Selection Problem}

Consider a simple chain graph with \(n\) nodes where node \(i\) is connected to node \(i+1\) with weight 1:

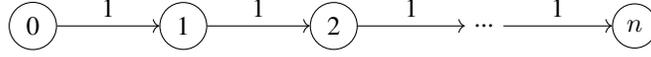
\begin{figure}[h!]
  \centering
  \begin{tikzpicture}[node distance=2cm]
    \node[circle,draw] (0) {0};
    \node[circle,draw] (1) [right of=0] {1};
    \node[circle,draw] (2) [right of=1] {2};
    \node (dots) [right of=2] {...};
    \node[circle,draw] (n) [right of=dots] {\(n\)};
     
    \draw[->] (0) -- node[above] {1} (1);
    \draw[->] (1) -- node[above] {1} (2);
    \draw[->] (2) -- node[above] {1} (dots);
    \draw[->] (dots) -- node[above] {1} (n);
  \end{tikzpicture}
  \caption{Chain graph demonstrating inefficient state selection}
  \label{fig:chain-graph}
\end{figure}

\begin{claim}[Complexity Impact of State Selection]
For the chain graph with \(n\) nodes, naive round-robin state selection leads to \(O(n^2)\) complexity, while informed state selection achieves \(O(n)\) complexity. More generally, for arbitrary graphs, naive state selection can lead to an exponential number of steps~\cite{Garg.2020}, making informed selection essential for practical performance.
\end{claim}

For example: consider solving SSSP from node 0. With naive selection, the algorithm examines all \(n\) nodes at each step \(i\), even though only node \(i\) is actually forbidden. This gives \(n\) iterations times \(n\) examinations per iteration, for a total complexity of \(O(n^2)\). With informed selection, the algorithm examines only node \(i\)'s neighbors at each step. Since exactly one new node becomes forbidden per step, the total complexity is \(n\) iterations times \(O(1)\) examinations, or \(O(n)\) overall.

\paragraph{Problem-Specific State Selection}

Different LLP problems provide different opportunities for informed state selection. In graph problems (SSSP, BFS), when a node \(v\) is advanced, only the neighbors of \(v\) can become forbidden, so we only need to add these neighbors to the processing queue. In Stable Marriage, when a person \(p\) is matched, only the current matches of \(p\)'s preferred partners can become forbidden. In Job Scheduling, when a job \(j\) is advanced, only successor jobs whose prerequisites include \(j\) can become forbidden. In each case, the key observation is that the set of potentially affected indices is much smaller than the full solution vector, and the problem structure tells us exactly which indices to check next.

\subsubsection{Concurrency Challenges in the LLP Framework}

Partitioning work among threads in the LLP framework presents significant challenges, particularly regarding the correctness of concurrent updates to the global state. We now formally analyze why naive concurrent implementations can lead to incorrect results.

\paragraph{Analysis of Race Conditions}

Consider the graph from our earlier example, where two threads are attempting to update \( \C[3] \) concurrently:

%\begin{claim}[Incorrectness of Non-Atomic Updates]
%In a parallel LLP implementation without atomic updates, the algorithm may terminate with a non-final state even when all other vertices are fixed.
%\end{claim}

%\begin{e1}
%Consider the following execution sequence:
%\begin{enumerate}
%  \item Initial state: \( \C = [0, 2, 5, \infty] \)
%  \item Thread \( T_1 \) reads \( \C[3] = \infty \)
%  \item Thread \( T_2 \) reads \( \C[3] = \infty \)
%  \item \( T_1 \) computes distance via edge \( b \): \( d_1 = \C[0] + w_b = 0 + 3 = 3 \)
%  \item \( T_2 \) computes distance via edge \( d \): \( d_2 = \C[2] + w_d = 5 + 3 = 8 \)
%  \item \( T_1 \) writes \( \C[3] = 3 \)
%  \item \( T_2 \) overwrites with \( \C[3] = 8 \)
%\end{enumerate}

% The final state \( \C = [0, 2, 5, 8] \) is incorrect because:
% \begin{itemize}
%   \item The true shortest path to vertex 3 has length 3 (via edge \( b \))
%   \item All other vertices are non-forbidden, so no further updates will be triggered
%   \item The algorithm terminates with \( \C[3] = 8 \), which is not the final shortest path distance
% \end{itemize}
%%\end{e1}

\begin{claim}[Incorrectness of Non-Atomic Updates]
  In a parallel LLP implementation without atomic updates (or without proper synchronization),
  the algorithm may terminate with a non-final state even when all other vertices are fixed.
  \end{claim}
   
  \begin{proof}[Sketch of Problematic Execution]
  Consider a shortest-path scenario, where $\C[v]$ tracks the current best-known distance to vertex~$v$.
  Suppose two threads, $T_1$ and $T_2$, each attempt to relax the same vertex $3$. 
  Initially, $\C = [\,0,\,2,\,5,\,\infty\,]$. (That is, $\C[0] = 0$, $\C[1] = 2$, $\C[2] = 5$, and $\C[3] = \infty$.)
   
  \begin{enumerate}
  \item Thread $T_1$ reads $\C[3] = \infty$ and computes a new distance $d_1 = \C[0] + w_b = 0 + 3 = 3$ via edge $b$.
  \item Thread $T_2$ reads $\C[3] = \infty$ and computes a new distance $d_2 = \C[2] + w_d = 5 + 3 = 8$ via edge $d$.
  \item $T_1$ writes $\C[3] = 3$.
  \item \textbf{Race Overwrite:} $T_2$ \emph{overwrites} $\C[3]$ to $8$ (using its stale read).
  \end{enumerate}
   
  At completion, the algorithm's final state becomes $\C = [\,0,\,2,\,5,\,8\,]$.
  Since the algorithm sees no further forbidden elements, it terminates—even though $\C[3]$ is incorrect. 
  The correct shortest path would have been $\C[3]=3$, but that better update was lost due to the non-atomic overwrite by $T_2$.
   
  This violates the global \emph{monotonicity} property: once we discover a strictly better distance (3 vs.\ $\infty$), we must not revert to a worse value (8). 
  Non-atomic writes allowed $T_2$ to use a stale read of $\C[3]=\infty$, effectively undoing $T_1$'s progress. 
  Hence, the algorithm can terminate in a \emph{non-final} state that fails to satisfy the intended shortest-path condition.
  \end{proof}

\paragraph{Generalization to Other LLP Problems}

The race condition pattern manifests differently in other LLP problems implemented in our framework:

\begin{itemize}
  \item \textbf{Breadth First Search (BFS):}
    In the LLP formulation, \(\C[v]\) tracks the BFS level (hop distance) from the source,
    and a vertex is forbidden if \(\C[v] > \min_{(u,v) \in E}(\C[u] + 1)\).
    Consider a graph where vertex~3 is reachable from the source in one hop (via vertex~0)
    and also in two hops (via vertices 0--2--3):
    \begin{itemize}
      \item \( T_1 \) relaxes vertex~3 via the direct edge from vertex~0: reads \(\C[0] = 0\), computes level~1
      \item \( T_2 \) relaxes vertex~3 via the two-hop path: reads \(\C[2] = 1\), computes level~2
      \item \( T_1 \) writes \(\C[3] = 1\)
      \item \( T_2 \) overwrites \(\C[3] = 2\), reverting the correct value
      \item The algorithm terminates with an incorrect BFS level for vertex~3
    \end{itemize}
   
  \item \textbf{Stable Marriage Problem:}
    In the LLP formulation, \(\C[m]\) tracks man~\(m\)'s proposal index into his preference list;
    the index advances monotonically.
    A man is forbidden if his current target woman prefers another man who also targets her.
    Consider a scenario where man~\(m_1\) is forbidden and two threads attempt to advance his proposal index concurrently:
    \begin{itemize}
      \item Both \( T_1 \) and \( T_2 \) read \(\C[m_1] = k\) (current proposal index)
      \item \( T_1 \) advances \(\C[m_1]\) from \(k\) to \(k+1\)
      \item \( T_2 \), still using its stale read, overwrites \(\C[m_1] = k+1\) with \(k+1\) (benign) or, in a more complex interleaving, a second advance by \(T_1\) to \(k+2\) is overwritten back to \(k+1\)
      \item The lost advance means \(m_1\) appears to target a woman he has already been rejected by, and the algorithm may terminate with a non-stable matching
    \end{itemize}
   
  \item \textbf{Job Scheduling:}
    In the LLP formulation, \(\C[j]\) tracks the earliest completion time of job~\(j\) in a DAG,
    advancing monotonically as predecessor completion times are discovered.
    A job is forbidden if its recorded completion time is less than the maximum of its predecessors' completion times plus its own duration:
    \begin{itemize}
      \item Two threads both read \(\C[j] = 5\) (current completion time estimate)
      \item \( T_1 \) discovers that predecessor \(p_1\) has finished at time~8, computes \(\C[j] = 8 + t_j = 12\)
      \item \( T_2 \) discovers that predecessor \(p_2\) has finished at time~10, computes \(\C[j] = 10 + t_j = 14\)
      \item \( T_2 \) writes \(\C[j] = 14\), then \( T_1 \) overwrites with \(\C[j] = 12\)
      \item The algorithm may terminate with a completion time that violates the precedence constraint from \(p_2\)
    \end{itemize}
\end{itemize}

\paragraph{Requirements for Correct Implementation}

To ensure correctness, parallel LLP implementations must satisfy:

\begin{enumerate}
  \item \textbf{Atomic Updates:} All modifications to \( \C[i] \) must be atomic, typically using compare-and-swap (CAS) operations:
  \[
  \text{CAS}(\C[i], \text{old}, \text{advance}(\text{old}))
  \]
  where \(\text{advance}\) is the problem-specific update function (e.g., \(\min(\text{old}, \text{new})\) for SSSP, index increment for Stable Marriage, \(\max\) over predecessors for Job Scheduling).
   
  \item \textbf{Monotonicity Preservation:} Updates must maintain the lattice property---values may only advance in the lattice order \(\sqsubseteq\), never regress:
  \[
  \forall i, t_2 > t_1: \C[i]_{t_1} \sqsubseteq \C[i]_{t_2}
  \]
  For SSSP, \(\sqsubseteq\) corresponds to \(\geq\) (distances decrease toward the optimum); for Stable Marriage, it corresponds to \(\leq\) on proposal indices (indices only increase); for Job Scheduling, it corresponds to \(\leq\) on completion times (times only increase).
   
  \item \textbf{Progress Detection:} The algorithm must be able to detect when no further improvements are possible:
  \[
  \not\exists i: \text{is\_forbidden}(\C, i)
  \]
\end{enumerate}

These requirements ensure that even with concurrent updates:
\begin{itemize}
  \item The algorithm will converge to the correct final state
  \item No valid updates are permanently lost
  \item The termination condition correctly identifies when the solution is complete
\end{itemize}

This concrete example highlights how the LLP framework leverages parallelism while also dealing with challenges such as the vast search space and the need for atomic operations to ensure correctness. It sets the stage for our discussion on practical implementations that address these challenges.

\subsection{From Theory to Practice}
\label{sec:cm-llp-practice}

The LLP model gives the theoretical foundation; this subsection shows how we map it to reusable systems code. We begin with the shared global state abstraction:

\begin{figure}[htb]
\begin{minipage}{\linewidth}
\begin{lstlisting}[language=rust, caption=Global state definition, basicstyle=\ttfamily\scriptsize]
pub struct GlobalState<T: Atomize> {
  pub solution_vector: Vec<T>,
  // ... additional fields ...
}
\end{lstlisting}
\end{minipage}
\end{figure}

This structure is generic over the type of the solution vector elements, which must implement the \texttt{Atomize} trait, meaning that it can both be read and written atomically and can be shared across threads.

Next, we define a core trait that encapsulates the essential operations needed for any LLP problem:

\begin{figure}[htb]
\begin{minipage}{\linewidth}
\begin{lstlisting}[language=rust, caption=Core trait definition for LLP problems, label={lst:core-llp-trait}, basicstyle=\ttfamily\scriptsize]
pub trait LatticeLinearProblem<T, S = T>: Send + Sync
where
  T: Atomize,
{
  /// Initializes and returns the global state
  fn init_global_state(&self) -> GlobalState<T>;

  /// Pushes initial states to process to the worklist
  fn initial_states_to_process<W: Worklist>(
    &self, 
    global_state: &GlobalState<T>, 
    worklist: &W
  );

  /// Returns whether an element is forbidden
  fn is_forbidden(
    &self, 
    global_state: &GlobalState<T>, 
    idx: usize
  ) -> bool;

  /// Advances a forbidden state
  fn advance<W: Worklist>(
    &self,
    global_state: &GlobalState<T>,
    idx: usize,
    worklist: &W,
  ) -> bool;

  /// Tests whether 'index' is forbidden and advances it if so
  /// A shorthand for calling is_forbidden and advance in sequence.
  fn ensure<W: Worklist>(
    &self,
    global_state: &GlobalState<T>,
    index: usize,
    worklist: &W,
  ) -> bool;

  /// Produces the final solution from the global state
  fn final_solution(&self, global_state: GlobalState<T>) -> Vec<S>;

  // ... additional helper methods and default implementations ...
}
\end{lstlisting}
\end{minipage}
\end{figure}

This trait defines the core operations that map directly to the theoretical concepts introduced earlier. The \texttt{init\_global\_state} method initializes the solution vector \(\C\) to its starting state, and \texttt{initial\_states\_to\_process} populates the worklist with the initial set of indices that need to be checked for forbiddenness. The \texttt{is\_forbidden} method implements the forbiddenness check \(\forbidden(\C,j,B)\) from Definition~\ref{def:forbidden}, while \texttt{advance} updates a forbidden state and returns whether further processing is needed---the worklist parameter allows the implementation to suggest which states should be processed next. The \texttt{ensure} method is a convenience that combines forbiddenness checking and advancement in a single call, with a default implementation that can be optimized for specific problems. Finally, \texttt{final\_solution} extracts the result from the global state (e.g.\ converting a vector of \texttt{AtomicU32} into a vector of \texttt{u32}).

\subsubsection{Single-Threaded Solvers}

We start with single-threaded solvers to illustrate the runtime interface before introducing concurrency. The simplest implementation is:

\begin{figure}[htb]
\begin{minipage}{\linewidth}
\begin{lstlisting}[language=rust,caption=A simple cyclic solver, basicstyle=\ttfamily\scriptsize]
pub fn lattice_linear_solver_cyclic_single_thread<T, S>(
  problem: &impl LatticeLinearProblem<T, S>,
) -> Vec<S>
where T: Atomize, S: Atomize {
  let global_state = problem.init_global_state();
  let mut worklist = NullWorklist {};

  loop {
    let mut found_forbidden = false;
    for idx in 0..global_state.solution_vector.len() {
      match problem.ensure(&global_state, idx, &mut worklist) {
        true => found_forbidden = true,
        false => (),
      }
    }
    if !found_forbidden {
      break;
    }
  }
  problem.final_solution(global_state)
}
\end{lstlisting}
\end{minipage}
\end{figure}

This solver implements the simplest, but also most inefficient, approach to finding a solution: repeatedly scan through all indices until no forbidden states are found. It uses the \texttt{ensure} method which combines forbiddenness checking and advancement in a single call, and it employs a \texttt{NullWorklist} since we are scanning through all indices anyway and do not need the advance step to suggest what to process next. The solver terminates only when a complete pass finds no forbidden states. Despite its simplicity, this implementation remains correct under LLP assumptions: monotone advances eventually reach a fixed point, and that fixed point corresponds to the intended solution when one exists.

While the cyclic solver is correct, it wastes effort checking states that are unlikely to be forbidden. A more efficient approach is to maintain a collection of states that might have become forbidden due to recent updates:

\begin{figure}[htb]
\begin{minipage}{\linewidth}
\begin{lstlisting}[language=rust,caption=A simple bag-based solver, basicstyle=\ttfamily\scriptsize]
pub fn lattice_linear_solver_bag_single_thread<T, S>(
  problem: &impl LatticeLinearProblem<T, S>,
) -> Vec<S>
where T: Atomize, S: Atomize {
  let global_state = problem.init_global_state();
  let bag = UnsafeCell::new(Vec::new());
   
  // Initialize with states that might be forbidden
  problem.initial_states_to_process(&global_state, &bag);

  // While there are states in the bag, process them.
  while let Some((idx, prio)) = bag.pop() {
    problem.ensure(&global_state, idx, &bag);
  }
  problem.final_solution(global_state)
}
\end{lstlisting}
\end{minipage}
\end{figure}

This implementation improves upon the cyclic solver by maintaining a bag of states that might be forbidden, rather than scanning all states on every pass. The \texttt{advance} method can add potentially affected states to the bag when it updates a state, and processing continues until the bag is empty.

The bag-based approach exploits the locality principle that is common to many LLP problems: when a state is advanced, only certain other states might become forbidden as a result. In SSSP, only the neighbors of an updated vertex need checking; in Stable Marriage, only the partners affected by a new match; in Job Scheduling, only successor jobs whose prerequisites include the settled job. By tracking only these potentially affected states, the solver avoids the wasted work of the cyclic approach while still remaining single-threaded. This implementation provides the foundation for the parallel versions that follow, since the worklist abstraction that drives it is the same one used by the multi-threaded solvers.

\subsubsection{Multi-Threaded Solvers}

We now present multi-threaded solvers that take full advantage of the parallel potential of the LLP framework.

\paragraph{All-Indices Parallel Solver:}
Rather than maintaining a dynamic worklist, this solver scans \emph{all} indices in parallel, repeatedly, until no forbidden index can be advanced. Although simpler, this approach can be quite effective if nearly every index frequently becomes forbidden, or if the overhead of using a dynamic work queue might overshadow its benefits.

\begin{figure}[htb]
\begin{minipage}{\linewidth}
 \begin{lstlisting}[language=rust,caption=Multi-threaded solver that scans all indices on every iteration,basicstyle=\ttfamily\scriptsize]
pub fn lattice_linear_solver_all_indices_parallel<T, S>(
  problem: &impl LatticeLinearProblem<T, S>,
  num_threads: usize,
) -> Vec<S>
where
  T: Atomize,
  S: Debug,
{
  // Build the global state (shared among threads)
  let global_state = problem.init_global_state();

  // Create a thread pool with exactly num_threads
  let pool = rayon::ThreadPoolBuilder::new()
    .num_threads(num_threads)
    .build()
    .expect("Failed building thread pool");

  // We repeat scans until no changes are made in two consecutive passes
  let mut loops_with_none_changed = 0;

  loop {
    let changed_any = AtomicBool::new(false);

    // In parallel, iterate over all indices
    pool.install(|| {
      (0..problem.size()).into_par_iter().for_each(|idx| {
        // If idx is forbidden, attempt to advance it
        if problem.is_forbidden(&global_state, idx) {
          let updated = problem.advance(&global_state, idx, &crate::NullWorklist {});
          if updated {
            changed_any.store(true, Ordering::Relaxed);
          }
        }
      });
    });

    // If no element was updated, we do a second check; if it still doesn't update,
    // then the algorithm terminates
    if !changed_any.load(Ordering::Relaxed) {
      loops_with_none_changed += 1;
      if loops_with_none_changed > 1 {
        break;
      }
    } else {
      loops_with_none_changed = 0;
    }
  }

  problem.final_solution(global_state)
}
\end{lstlisting}
\end{minipage}
\end{figure}

Compared to the other multi-threaded solvers in this framework, the all-indices-parallel approach offers \emph{simplicity}, since it requires no worklist and no logic to select which indices to process next; it simply scans them all. This can, however, introduce \emph{potential overhead} in large problems, because many indices may already be fixed at any given time, resulting in unproductive re-checks. Nevertheless, it can be advantageous when the problem structure causes many (or all) indices to frequently become forbidden, or when the solution vector itself is relatively small—meaning the primary cost lies in evaluating predicate conditions rather than deciding which indices to examine. Crucially, \emph{correctness} is preserved by the LLP model: scanning every index repeatedly will still converge to a valid solution if one exists, just as bag-based approaches do.

\paragraph{Shared Work Bag:}

This version is a simple extension of the bag-based single-threaded solver. It uses a shared work bag (a collection of elements that might be forbidden with no particular order) to distribute work across multiple threads:

\begin{figure}[htb]
\begin{minipage}{\linewidth}
 \begin{lstlisting}[language=rust,caption=Simplified multi-threaded LLP solver with work stealing,basicstyle=\ttfamily\scriptsize]
  pub fn lattice_linear_solver_multi_threaded_shared_work_bag<T, S>(
    problem: &impl LatticeLinearProblem<T, S>,
    num_threads: usize,
  ) -> Vec<S>
  where T: Atomize, S: Atomize {
    // Initialize shared state
    let global_state = Arc::new(problem.init_global_state());
    let shared_bag = Injector::new(); // Lock-free work-stealing queue
     
    // Initialize with starting states
    problem.initial_states_to_process(&global_state, &shared_bag);
   
    // Create thread pool and process elements
    rayon::ThreadPoolBuilder::new()
      .num_threads(num_threads)
      .build()
      .unwrap()
      .install(|| {
        (0..num_threads).into_par_iter().for_each(|_| {
          // Each thread processes work until no more work is available
          while let Steal::Success((idx, _)) = shared_bag.steal() {
            if problem.ensure(&global_state, idx, &shared_bag) {
              // State was forbidden and advanced
              trace_forbidden(idx, true);
            }
          }
        });
      });
   
    problem.final_solution(Arc::try_unwrap(global_state).unwrap())
  }
\end{lstlisting}
\end{minipage}
\end{figure}

For this purpose, the implementation uses \texttt{Arc}, which stands for Atomic Reference Counting, for the distribution of global state between threads. This allows for the concurrent access of the solution vector through atomic operations without the need for locks. This allows each thread to read and modify the state independently, hence reducing the number of contention points.

The work-stealing queue (\texttt{Injector}) distributes work across threads. When a thread advances a state, it adds newly discovered forbidden states to this shared queue. Threads that have no work can steal tasks from the queue, automatically balancing the workload across the thread pool.

Rayon's thread pool handles thread creation and management. Each thread runs until it can no longer steal work from the shared queue. The LLP framework's monotonic property ensures correctness: states can only advance a finite number of times, and each advancement moves toward the final solution regardless of the order of processing.

This approach ensures that the thread making progress through a certain state enqueues the states that are potentially forbidden, which may be processed by any thread, hence ensuring that there is always progress and that all threads are utilized almost uniformly.

Although this approach is robust, it may be prone to performance degradation at high thread counts, hence the need for using thread queues.

\paragraph{Per-Thread Work Bags:}
To reduce contention from multiple threads accessing a shared work bag, we implement a version with per-thread work bags. Each thread maintains its own local work queue, only attempting to steal work from other threads when its queue is empty.

\begin{figure}[htb]
\begin{minipage}{\linewidth}
\begin{lstlisting}[language=rust,caption=Simplified version of a multi-threaded LLP solver with per-thread work bags,basicstyle=\ttfamily\scriptsize]
pub fn lattice_linear_solver_multi_threaded_per_thread_work_bag<T, S>(
  problem: &impl LatticeLinearProblem<T, S>,
  num_threads: usize,
) -> Vec<S>
where T: Atomize + Sync, S: Atomize {
  let global_state = Arc::new(problem.init_global_state());

  // Create a global injector and per-thread workers
  let injector = Injector::new();
  let workers: Vec<Worker<(usize, usize)>> = 
    (0..num_threads).map(|_| Worker::new_fifo()).collect();
  let stealers: Vec<Stealer<(usize, usize)>> = 
    workers.iter().map(|w| w.stealer()).collect();

  // Initialize with starting states
  problem.initial_states_to_process(&global_state, &injector);

  thread_pool.scope(|s| {
    for worker in workers.into_iter() {
      s.spawn(move |_| {
        let mut worker = worker;
        loop {
          // Try local queue first, then global, then steal from others
          let task = worker
            .pop()
            .or_else(|| injector.steal().success())
            .or_else(|| {
              for stealer in &stealers {
                if let Steal::Success(t) = stealer.steal() {
                  return Some(t);
                }
              }
              None
            });

          match task {
            Some((idx, _)) => {
              if problem.is_forbidden(&global_state, idx) {
                problem.advance(&global_state, idx, &mut worker);
              }
            }
            None => break, // No more work available
          }
        }
      });
    }
  });

  problem.final_solution(Arc::try_unwrap(global_state).unwrap())
}
\end{lstlisting}
\end{minipage}
\end{figure}

This implementation introduces a hierarchical work distribution system. At its core, each thread operates on a dedicated, thread-local bag, implemented as a \texttt{Worker}. These local bags serve as the primary source of work for each thread, allowing most operations to proceed without any inter-thread synchronization. The global \texttt{Injector} queue holds the initial work items and serves as a secondary source of work.

The work-stealing mechanism follows a priority order. A thread first attempts to retrieve work from its local queue, maximizing cache locality and minimizing synchronization overhead. When the local queue is empty, the thread attempts to steal work from the global injector. Only after both local and global sources are exhausted does the thread attempt to steal work from other threads' queues, scanning through available stealers until work is found or all possibilities are exhausted.

This hierarchical approach significantly reduces contention compared to the shared work bag implementation. Threads primarily operate on their local queues without synchronization, only engaging in potentially contentious operations when local work is depleted. The design benefits problems with good locality characteristics, where related work items tend to be processed by the same thread.

The implementation maintains the same correctness guarantees as the shared work bag version through the LLP framework's monotonic properties. When a thread advances a state, new potentially forbidden states are added to its local queue, maintaining locality while still allowing work distribution through the stealing mechanism when necessary for load balancing.

\paragraph{Parallel Buckets Solver:}
In some problems, a natural priority (or cost) metric emerges for each element, such as a distance or a weight. In these scenarios, we can group elements into a series of \emph{buckets}, each corresponding to a range of priority values, and allow threads to ``pop'' from the lowest non-empty bucket first. This approach helps the algorithm focus on states with higher urgency (or lower cost) and can often yield significant speedups.

\begin{figure}[htb]
\begin{minipage}{\linewidth}
\begin{lstlisting}[language=rust,caption=Multi-threaded solver that uses buckets to organize work,basicstyle=\ttfamily\scriptsize]
pub fn lattice_linear_solver_multi_threaded_buckets<T, S, AdditionalState>(
  problem: &impl LatticeLinearProblem<T, S, AdditionalState>,
  num_threads: usize,
  delta: usize,
) -> Vec<S>
where
  T: Atomize,
  S: Debug,
  AdditionalState: Debug + Send + Sync,
{
  // 1) Initialize global state in an Arc, so threads can share it.
  let global_state = Arc::new(problem.init_global_state());

  // 2) Create a bucket-based worklist (also in an Arc).
  let num_buckets = 1024;
  let bucket_worklist = Arc::new(BucketWorklist::new(num_buckets, delta));

  // 3) Enqueue initial tasks from the problem.
  problem.initial_states_to_process(&global_state, &*bucket_worklist);

  // 4) Build a thread pool with num_threads.
  let thread_pool = rayon::ThreadPoolBuilder::new()
    .num_threads(num_threads)
    .build()
    .unwrap();

  // 5) The concurrency pattern: each of the num_threads workers
  //  repeatedly pops from the bucketed worklist until no more items remain.
  thread_pool.install(|| {
    (0..num_threads).into_par_iter().for_each(|_| {
      while let Some((idx, cost)) = bucket_worklist.pop() {
        if problem.is_forbidden(&global_state, idx) {
          problem.advance(&global_state, idx, &*bucket_worklist);
        }
      }
    });
  });

  // 6) Extract the final solution from Arc once all threads have finished.
  let global_state = Arc::try_unwrap(global_state).unwrap();
  problem.final_solution(global_state)
}
\end{lstlisting}
\end{minipage}
\end{figure}

In contrast to per-thread or shared-bag solvers, this \textbf{bucket-based} approach is most useful when each index (vertex, edge, or job) has a meaningful priority. Grouping work items by cost ranges lets the solver process lower-cost updates first. This can reduce redundant work in problems such as SSSP, where high-distance updates are often less urgent.

Each thread repeatedly pops tasks from the lowest non-empty bucket, advances the corresponding state if needed, and pushes newly forbidden states (with their costs) back into the appropriate bucket. As in the other solvers, correctness follows from monotonic atomic updates: once no forbidden indices remain, the computation converges.

\subsubsection{Comparison of Multi-threaded Solvers}

Different multi-threaded solvers are preferable in different contexts. Key selection criteria include:

\begin{itemize}
  \item \textbf{Solution vector size:} For smaller vectors (for example, \(10\) to \(10^4\) elements), all-indices-parallel scanning can be acceptable. For larger vectors (for example, SSSP on millions of vertices), targeted worklists are usually necessary.
  \item \textbf{Predicate locality:} If forbiddenness changes remain local, per-thread work bags often improve cache locality and reduce stealing. If updates can activate distant states, shared work bags may distribute work more effectively.
  \item \textbf{Priority structure:} Some problems expose natural priorities (for example, distance in SSSP). When such priorities are informative, bucketed scheduling can reduce wasted updates.
\end{itemize}

One key advantage of LLP-FW is the separation between problem logic and solver logic. As long as a problem implementation satisfies the framework interface, we can evaluate multiple solver variants without rewriting the core algorithm. This makes solver choice an empirical decision per workload rather than a one-time architectural commitment.

\subsection{Memoization Through Fixed States}
\label{sec:cm-memoization}

Some LLP problems expose an additional optimization: once an index reaches a problem-specific \emph{fixed} condition, future monotone updates cannot make it forbidden again. In those cases, rechecking that index is redundant work. LLP-FW therefore supports memoization through a lock-free fixed-state bit vector in the global state.

The optimization is conservative. A problem marks an index as fixed only when it can justify that decision from its own invariants (for example, a shortest-path label that is provably final under the active schedule policy). Solvers then skip fixed indices in hot loops and can use an ``all fixed'' fast path for termination checks. This removes repeated forbiddenness tests, reduces queue churn, and lowers contention on shared structures.

The benefit is empirical rather than universal. In workloads where labels stabilize early, memoization noticeably reduces wasted work; where values keep changing, the fixed set grows slowly and the gain is smaller. The trends in the SSSP and Job Scheduling scheduler analyses (Figures~\ref{fig:sssp-worklists} and~\ref{fig:job-worklists}) are consistent with this mechanism: policies that concentrate work near recent updates allow more indices to become fixed earlier and stay out of the active frontier.

\section{Common LLP problems implemented in the LLP framework}

This section presents seven LLP workloads using one common template: \emph{Formulation}, \emph{Implementation}, \emph{Setup}, and \emph{Results and Discussion}. We keep exploratory plots only when they materially change interpretation (for example, revealing scheduler sensitivity or regime changes); broader raw sweeps remain available in the generated tables and auxiliary artifacts.

\subsection{Shared Experimental Protocol}
\label{sec:cm-eval}

We implemented LLP-FW in Rust using custom lock-free and concurrent data structures, with \texttt{rayon} for selected runtime components (for example, work stealing). Benchmarks are implemented and executed with the \texttt{criterion} rust crate.

We ran benchmarks on an AWS virtual machine of type: \texttt{c7i.16xlarge} VM (Intel Xeon Platinum 8375C, 2.3 GHz, 128 GB RAM, NVMe storage).

Unless explicitly noted otherwise, we report medians over repeated runs, category summaries use geometric means, and speedups are computed from per-thread medians against the named baseline in each figure/table. Differences below roughly 5\% are interpreted as near parity unless they repeat consistently across datasets.

\paragraph{Threats to validity:}
Short-running workloads on shared-cloud hardware can be affected by background jitter, turbo variation, and NUMA placement. Baselines are tuned but constrained by implementation choices available in this codebase (queue structures, \(\Delta\)-stepping parameters, bucket hints), so speedups are in comparison to the internal implementations only. We made this choice to ensure that results were directly comparable and not affected by different runtimes or programming languages.

\subsection{Single Source Shortest Path}   
\label{sec:llp-sssp}

The Single Source Shortest Path (SSSP) problem involves finding the shortest paths from a designated source vertex to all other vertices in a weighted graph. It is a fundamental problem in graph theory and has wide applications in networks, routing, and optimization.

\noindent\textbf{Problem Definition:} Given a weighted graph \( G = (V, E, w) \) where \( V \) is the set of vertices, \( E \) is the set of edges, and \( w: E \rightarrow \mathbb{R}^+ \) assigns a non-negative weight to each edge, and a source vertex \( s \in V \), find the shortest path from \( s \) to every other vertex \( v \in V \).

\noindent The goal is to compute a distance function \( d: V \rightarrow \mathbb{R}^+ \cup \{\infty\} \) such that \( d(s) = 0 \) and for all \( v \in V \), \( d(v) \) is the minimum total weight of any path from \( s \) to \( v \).

\subsubsection{Formulation}
\label{sec:sssp-llp-formulation}

The SSSP problem can be framed as a lattice linear problem by defining a suitable lattice structure and forbiddenness conditions.

\noindent\textbf{Lattice Structure:} The solution space is a lattice where each vertex \( v \) is associated with a distance \( d(v) \). The lattice ordering is defined by the natural ordering of distances: for distances \( d_1 \) and \( d_2 \), \( d_1 \leq d_2 \) if \( d_1 \) is less than or equal to \( d_2 \).

\noindent\textbf{Global State:} The global state comprises the distance estimates for all vertices:

\[
\text{GlobalState} = \{ d(v) \mid v \in V \}
\]

\noindent Initially, \( d(s) = 0 \) and \( d(v) = \infty \) for all \( v \neq s \).

\noindent\textbf{Forbiddenness Condition:} A vertex \( v \) is considered \emph{forbidden} if there exists an edge \( (u, v) \in E \) such that:

\[
d(v) > d(u) + w(u, v)
\]

\noindent This implies that there is a shorter path to \( v \) via \( u \) that has not yet been accounted for in \( d(v) \).

\smallskip

\noindent\textbf{Advance Function:} When a vertex \( v \) is forbidden, the advance function updates its distance:

\[
d(v) := \min \{ d(v), d(u) + w(u, v) \mid (u, v) \in E \}
\]

\noindent After updating, vertices that depend on \( v \) (i.e., neighbors of \( v \)) may become forbidden and need to be processed.

This formulation is closely related to the classical Bellman-Ford algorithm: each relaxation step in Bellman-Ford corresponds exactly to advancing a forbidden vertex in the LLP view. The difference is in how work is scheduled. Bellman-Ford processes all vertices in lock-step rounds; $\Delta$-stepping relaxes them in global buckets grouped by tentative distance. The LLP formulation instead makes the scheduling decision purely local: any vertex that violates the forbiddenness predicate can be advanced by any thread at any time, with no global round or bucket barrier. This means the LLP runtime can exploit fine-grained parallelism that round-based approaches leave on the table, at the cost of additional per-vertex coordination (atomic reads and compare-and-swap updates). Whether that trade-off pays off in practice depends on how many vertices are simultaneously forbidden and how effectively the worklist keeps threads focused on productive work (see the example in Section~\ref{sec:cm-llp-example}).

\subsubsection{Implementation}
\label{sec:sssp-impl}

To demonstrate how the LatticeLinearProblem trait presented in Listing~\ref{lst:core-llp-trait} materialises in practice, we examine a simplified version of the Single-Source Shortest Path implementation. The problem is captured by a compressed sparse row graph together with the distinguished source vertex (Listing~\ref{lst:sssp-code-struct}). Within the LLP framework three entities interact: the immutable problem definition (\texttt{struct SSSP}), the shared mutable state (\texttt{struct GlobalState}) that threads update through atomics, and the worklist abstraction (\texttt{struct Worklist}) that steers the solver toward the most promising vertices.

\begin{figure}[htb]
\begin{minipage}{\linewidth}
\begin{lstlisting}[language=rust,caption=Representation of the SSSP problem, label={lst:sssp-code-struct}, basicstyle=\ttfamily\scriptsize]
pub struct SSSP {
  // The graph, stored as compressed sparse rows
  graph: &UndirectedCsrGraph<usize, (), usize>,
  // The source node
  source: usize,
}

impl SSSP {
  // Return the minimum distance from any neighbor
  fn min_dist_from_neighbors(&self,
    global_state: &GlobalState<usize>,
    index: usize) -> usize;

  // Compare-and-swap the distance for a given vertex,
  // retuning the actual distance if the compare-and-swap fails.
  fn compare_and_swap(&self,
            global_state: &GlobalState<usize>,
            index: usize,
            new_dist: usize) -> Result<(), usize>;
}

{...}
\end{lstlisting}
\end{minipage}
\end{figure}

We then implement the LatticeLinearProblem trait for the SSSP problem, including the \texttt{is\_forbidden} and \texttt{advance} functions. We chose the LLP Bellman-Ford as it is easy to explain and implement while still providing quite good performance in practice. As the actual implementation quite more complex a simplified version is shown in Listing~\ref{lst:sssp-code-impl}.

\begin{figure}[htb]
\begin{minipage}{\linewidth}
\begin{lstlisting}[language=rust,caption=Simplified implementation of the SSSP problem, label={lst:sssp-code-impl}, basicstyle=\ttfamily\scriptsize]
impl LatticeLinearProblem<usize> for SSSP {
  fn is_forbidden(&self, global_state: &GlobalState<usize>, index: usize) -> bool {
    // Source vertex is never forbidden
    if index == self.source {
      return false;
    }
    let current_dist = global_state.read(index);   
    // Check if any neighbor offers a shorter path
    self.min_dist_from_neighbors(global_state, index) < current_dist
  }

  fn advance<W: Worklist>(
    &self,
    global_state: &GlobalState<usize>,
    index: usize,
    worklist: &W,
  ) -> bool {
    let mut current_dist = global_state.read(index);
     
    // Find minimum distance through any neighbor and update the distance if it's smaller,
    // using compare-and-swap to ensure atomicity.
    loop {
      let min_dist = self.min_dist_from_neighbors(global_state, index);
      if min_dist >= current_dist {
        return false;
      }

      match self.compare_and_swap(global_state, index, min_dist) {
        // If the compare-and-swap succeeds, we update the distance and add the neighbors
        // to the worklist.
        Ok(_) => {
          worklist
            .push_all(self.graph.neighbors(index));
          return true;
        }
        Err(actual) => current_dist = actual,
      }
    }
    false
  }
}
\end{lstlisting}
\end{minipage}
\end{figure}

The full implementation is in \texttt{src/algorithms/sssp.rs}, but the above listing gives a sense of the essential structure. There are two features that set this apart from a standard textbook implementation of the Bellman Ford loop. One is that the update of the distance in the \texttt{advance} loop is done using a compare and swap retry loop, not a guarded write, which allows us to relax the distance of a vertex concurrently from many threads without the need to lock the vertex. The other is that, after a successful relaxation, only the neighbours of the vertex are added to the work list, which is a demand-driven approach that does not scan the entire set of vertices on each iteration.

The full solver also adds a significant scheduling refinement to this basic loop structure. Whenever the distance to a vertex is reduced, the neighbours of the vertex are enqueued to the per-thread work bag scheduler with a priority based on the new tentative distance. The idea is that a neighbour of a vertex that has just had its distance reduced is likely to become forbidden shortly, and therefore it is a waste to relax it later, as it would cause redundant work if it is eventually reached by a longer path. This "recency bias" focuses each thread's work near the currently active frontier and minimizes the number of compare-and-swap operations performed by the solver. Section~\ref{sec:sssp-results} presents the results, showing that PTWB can reduce runtime by an order of magnitude compared to schedulers without this priority cue on our benchmark graphs.

\subsubsection{Setup}

We use a variety of synthetic and real-world graphs to evaluate our SSSP solver, including power law, road, citation, DAG, and mesh graphs. Some representative datasets and size ranges are listed in Table~\ref{tab:sssp-setup}. Our baseline is a Rust implementation of the $\Delta$-stepping algorithm with a shared bucket data structure, using a tuned $\Delta$ value, which is determined by selecting the appropriate power of two on the median edge weight across each graph family using a grid search. We use a fixed random seed to select source vertices, holding this constant across all solvers and thread counts for a particular graph. All experiments are performed on the platform described in Section~\ref{sec:cm-eval}.

\begin{table}[htb]
  \centering
  \resizebox{\linewidth}{!}{
  \begin{tabular}{lccc}
    \toprule
    \textbf{Graph family} & \textbf{Representative datasets} & \textbf{\(|V|\)} & \textbf{\(|E|\)} \\
    \midrule
    Power-law (Kronecker/RMAT) & \texttt{kron-2e14}..\texttt{kron-2e18}, \texttt{rmat8} & $1.6\times10^{4}$--$2.6\times10^{5}$ & $6.2\times10^{4}$--$1.5\times10^{6}$ \\
    Road network & \texttt{Asia\_rand} & $1.2\times10^{7}$ & $1.27\times10^{7}$ \\
    Citation networks & \texttt{coAuthorsCiteseer}, \texttt{citationCiteseer}, \texttt{coAuthorsDBLP} & $2.3\times10^{5}$--$3.0\times10^{5}$ & $8.1\times10^{5}$--$1.2\times10^{6}$ \\
    DAGs & \texttt{DAG1}..\texttt{DAG3} & $5.0\times10^{2}$--$2.0\times10^{3}$ & $1.8\times10^{4}$--$4.8\times10^{4}$ \\
    Structured meshes & \texttt{torus5}, \texttt{rome99} & $3.2\times10^{1}$--$3.4\times10^{3}$ & $9.2\times10^{1}$--$7.1\times10^{3}$ \\
    Synthetic power-law (graph500) & \texttt{graph500-s18-ef16} & $1.7\times10^{5}$ & $7.6\times10^{6}$ \\
    \bottomrule
  \end{tabular}
  }
  \caption{Graph families and example datasets used in the SSSP evaluation.}
  \label{tab:sssp-setup}
\end{table}

\subsubsection{Results and Discussion}
\label{sec:sssp-results}

We structure the SSSP evaluation around four questions: what overhead does the LLP abstraction add at low thread counts, where does it pay off as parallelism increases, how sensitive is performance to worklist policy, and where does LLP fail to improve over the baseline.

\smallskip
\noindent\textbf{Single-thread overhead.}
At one thread, the LLP solver is consistently slower than $\Delta$-stepping. On \texttt{kron-2e14} the best single-thread LLP variant (\emph{all-parallel}) runs in \(3.91\)~ms versus \(0.38\)~ms for $\Delta$-stepping, roughly a \(10\times\) overhead. On the larger \texttt{kron-2e16} the ratio is similar: \(23.8\)~ms versus \(2.6\)~ms. This overhead has two sources. First, every distance update requires a compare-and-swap even when no other thread is contending, whereas $\Delta$-stepping performs plain writes in its sequential path. Second, the LLP solver tests the forbiddenness predicate for each vertex it pops, which involves reading all neighbour distances; $\Delta$-stepping avoids this by maintaining bucket membership implicitly. The overhead is the price of generality: the same code path works unchanged at higher thread counts where the atomic operations become necessary.

\smallskip
\noindent\textbf{Baseline comparison across graph morphologies.}
Table~\ref{tab:sssp-llp-baseline} reports the fastest LLP configuration versus tuned $\Delta$-stepping at 32 threads across representative graphs. The speedups are heterogeneous and instructive.

\begin{table}[htb]
  \centering
  \resizebox{\linewidth}{!}{
  \begin{tabular}{lrrrrr}
    \toprule
    \textbf{Dataset} & \textbf{$|V|$} & \textbf{$|E|$} & \textbf{$\Delta$-stepping (ms)} & \textbf{Best LLP (ms)} & \textbf{Speedup} \\
    \midrule
    \texttt{kron-2e14-rnd} & 16\,364 & 61\,851 & 11.88 & 2.51 & $4.74\times$ \\
    \texttt{kron-2e15-rnd} & 32\,752 & 136\,319 & 13.47 & 3.57 & $3.77\times$ \\
    \texttt{rmat8-2e14} & 16\,383 & 130\,148 & 10.93 & 3.25 & $3.37\times$ \\
    \texttt{kron-2e16-rnd} & 65\,523 & 300\,245 & 14.73 & 6.92 & $2.13\times$ \\
    \texttt{Asia\_rand} & 11\,950\,757 & 12\,711\,603 & 872.10 & 569.06 & $1.53\times$ \\
    \texttt{graph500-s18-ef16} & 174\,066 & 7\,600\,884 & 23.44 & 25.29 & $0.93\times$ \\
    \texttt{coAuthorsDBLP} & 299\,067 & 1\,207\,064 & 12.81 & 34.57 & $0.37\times$ \\
    \bottomrule
  \end{tabular}
  }
  \caption{32-thread SSSP comparison between tuned $\Delta$-stepping and the fastest LLP configuration per dataset.}
  \label{tab:sssp-llp-baseline}
\end{table}

The results reveal a clear density-dependent pattern. Sparse graphs with narrow frontiers favour LLP: \texttt{kron-2e14} reaches \(4.7\times\), and \texttt{rmat8-2e14} achieves \(3.4\times\) despite having \(2\times\) more edges at similar vertex count. As edge density grows further, the advantage shrinks: \texttt{kron-2e16} (\(300\)k edges, \(2.13\times\)) has roughly half the speedup of \texttt{kron-2e15} (\(136\)k edges, \(3.77\times\)), tracking the growth in frontier width. The road network \texttt{Asia\_rand} achieves a modest \(1.53\times\) despite its narrow frontiers, because the large graph diameter means both approaches must process many sequential wavefronts. On the dense end, the two losses have different causes but the same root. \texttt{graph500-s18-ef16} (\(0.93\times\)) has an average degree of \(\approx\!44\): each relaxation pushes dozens of neighbours onto the worklist, most of which are already being processed by other threads, so the frontier is persistently wide and LLP's per-vertex CAS cost is paid on nearly every vertex with no compensating selectivity. \texttt{coAuthorsDBLP} (\(0.37\times\)) has moderate average degree (\(\approx\!4\)) but strong community structure and small-world diameter, which causes the SSSP frontier to rapidly engulf a large fraction of the graph; vertices within a community settle at similar distances simultaneously, creating the same wide-frontier regime. In both cases $\Delta$-stepping's synchronized bucket sweeps amortise coordination over many vertices per phase using plain stores, while LLP pays the atomic overhead on every individual update.

\smallskip
\noindent\textbf{Scaling behavior.}
Figure~\ref{fig:sssp-scaling} reports strong scaling on \texttt{kron-2e16}. At one thread, $\Delta$-stepping is substantially faster due to the overhead discussed above. As thread count increases, $\Delta$-stepping's global bucket barrier becomes a bottleneck: threads must synchronize at the end of each bucket before the next distance range is released, and the useful work per round shrinks as the frontier narrows. The LLP solver has no such barrier. Each thread independently pops forbidden vertices from its local queue and attempts relaxations, so work is distributed continuously rather than in synchronized waves. The crossover occurs around four threads in this experiment, and the LLP advantage widens modestly through eight threads. On the road network \texttt{Asia\_rand}, the pattern differs: the baseline's global bucket barrier produces erratic scaling, while LLP-PTWB improves steadily from \(10{,}297\)~ms at one thread to \(569\)~ms at 32 threads. The crossover occurs around four threads (\(1.81\times\)), and LLP reaches \(1.53\times\) at 32 threads---a modest win, but one that reflects the road graph's narrow frontiers favouring per-thread locality even when the absolute speedup over the baseline is limited by the smaller graph diameter relative to the BFS case.

\begin{figure}[htb]
  \centering
  \includegraphics[width=0.48\linewidth]{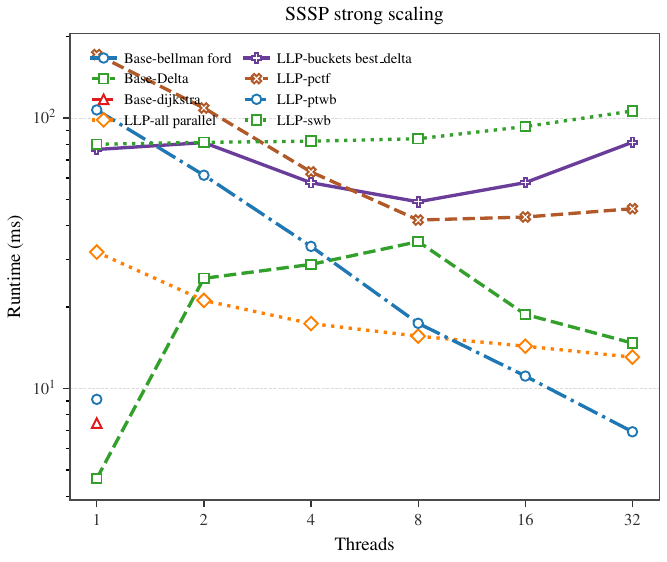}\hfill
  \includegraphics[width=0.48\linewidth]{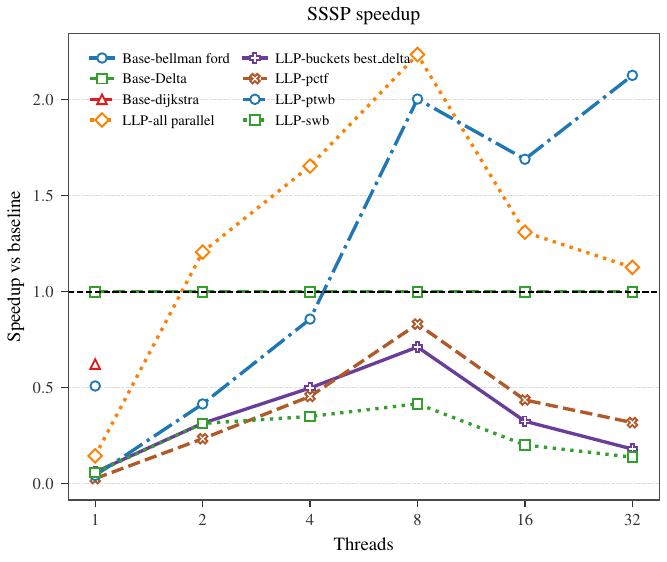}
  \caption{Runtime (left) and speedup over $\Delta$-stepping (right) on \texttt{kron-2e16}.}
  \label{fig:sssp-scaling}
\end{figure}

\smallskip
\noindent\textbf{Worklist policy is a first-order effect.}
Figure~\ref{fig:sssp-worklists} isolates scheduler effects by holding the graph fixed at \texttt{kron-2e16} and varying only the LLP worklist at eight threads. The difference is dramatic: PTWB finishes in roughly \(11.6\)~ms, while SWB requires \(40.9\)~ms and bucketed queues need \(79.1\)~ms. The mechanism behind PTWB's advantage is the ``recency bias'' introduced in Section~\ref{sec:sssp-impl}: when a vertex's distance decreases, its neighbours are enqueued with a priority reflecting the new tentative distance. Threads therefore revisit vertices near the active frontier before stale vertices deeper in the queue. This reduces redundant relaxations---cases where a vertex is popped from the worklist but its distance has already been improved by another thread since it was enqueued. SWB lacks this priority signal entirely, and bucketed queues group by distance range but do not distinguish between fresh and stale entries within a bucket. On the road network at 32 threads the effect is even larger: PTWB reaches \(293.7\)~ms while the next-best policy (PTCF) requires \(2{,}541\)~ms, an \(8.6\times\) gap.

\begin{figure}[htb]
  \centering
  \includegraphics[width=0.6\linewidth]{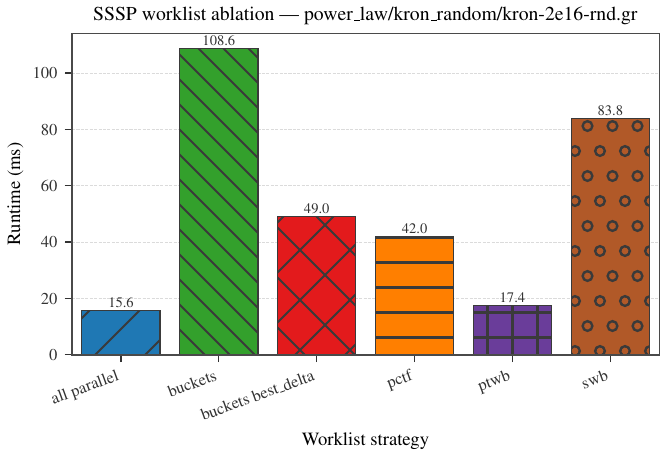}
  \caption{Impact of LLP worklists on \texttt{kron-2e16} at eight threads. PTWB minimises redundant relaxations.}
  \label{fig:sssp-worklists}
\end{figure}

\smallskip
\noindent\textbf{Where LLP does not help.}
The LLP approach does not uniformly dominate. On the citation networks in our benchmark suite (\texttt{coAuthorsDBLP}, \texttt{citationCiteseer}), $\Delta$-stepping remains faster even at 32 threads. These graphs have moderate vertex counts (\(\approx\)300k) but dense, highly connected community structure that produces wide forbidden frontiers at nearly every step. In this regime the per-vertex atomic overhead is paid on a large fraction of vertices simultaneously, and the priority signal that PTWB relies on is less effective because many vertices improve in parallel with similar tentative distances, diluting the benefit of recency-based ordering. This limitation is a useful boundary marker: the LLP formulation is most effective when the forbidden frontier is narrow relative to the graph, and less effective when the frontier is persistently wide.

\smallskip
\noindent\textbf{Summary.}
The SSSP evaluation reveals a clear trade-off profile for the LLP approach. At one thread, the abstraction adds measurable overhead (\(\approx\!\!10\times\)) due to atomic operations and predicate checks. That overhead is recovered as thread counts grow, with the crossover point depending on graph morphology: sparse or structured graphs (power-law Kronecker, road networks) cross over early and deliver large gains at high thread counts, while dense community-structured graphs remain unfavorable. Across all regimes, worklist policy is at least as important as raw parallelism: PTWB's recency-biased scheduling reduces redundant work by an order of magnitude on graphs where it applies. For a practitioner, the implication is that LLP is a strong choice for SSSP when the target graphs have narrow forbidden frontiers and the hardware offers moderate-to-high parallelism, but dense graphs where a large fraction of vertices are simultaneously forbidden are better served by bucket-based approaches like $\Delta$-stepping.

\subsection{Breadth-First Search (BFS)}
\label{sec:bfs}

\subsubsection{Formulation}
Breadth-First Search assigns each vertex a \emph{level}---its hop distance from a distinguished source---and produces a shortest-path tree for unweighted graphs. In the LLP view, the global state is a vector of level estimates, initially \(\infty\) everywhere except the source (level~0). A vertex \(v\) is \emph{forbidden} if it has a neighbour \(u\) whose level is strictly less than \(d(v) - 1\), meaning a shorter path to \(v\) exists through \(u\) that has not yet been recorded. Advancing a forbidden vertex sets its level to one plus the minimum neighbour level and enqueues its own neighbours for re-evaluation.

Because all edge weights are implicitly one, levels can only decrease and each vertex settles at most once: after a vertex is first discovered from the nearest frontier, no future advance can lower its level further. This monotonicity is the same lattice property that SSSP exploits, but with a simpler order (integer hops rather than real-valued distances). The practical consequence is that the forbidden frontier in BFS tends to be thinner and more structured than in weighted SSSP, which changes the relative importance of scheduler policy as we discuss below.

\subsubsection{Implementation}
The LLP implementation mirrors the familiar frontier-based BFS. Vertices are stored in contiguous arrays and levels are updated through atomic compare-and-swap, exactly as in the SSSP solver. The key difference is that no priority signal accompanies the worklist push: since all edges have unit weight, there is no distance delta to use as a scheduling hint. Instead, the solver relies on spatial locality---PTWB partitions the frontier across threads so that each thread processes a contiguous slab of neighbours, reducing cache thrashing and cross-core traffic. For graphs where the frontier becomes very wide (dense power-law instances), the solver can fall back to a shared work bag (SWB) that distributes vertices more evenly at the cost of a global queue.

\subsubsection{Setup}
We exercise BFS on a superset of the SSSP graph families, adding larger Kronecker instances (up to \texttt{kron-2e22}, 4.2M vertices, 34M edges) and Erd\H{o}s--R\'enyi random graphs. Table~\ref{tab:bfs-setup} summarises the datasets. The baseline is a multi-threaded lock-free queue BFS shipped with the benchmark suite (\texttt{baseline mt}), which uses a global concurrent queue for frontier management. Source vertices are chosen with a fixed random seed and held constant across all solvers. All runs use the platform described in Section~\ref{sec:cm-eval}.

\begin{table}[htb]
  \centering
  \small
  \begin{tabular}{lccc}
    \toprule
    \textbf{Graph family} & \textbf{Representative datasets} & \textbf{$|V|$} & \textbf{$|E|$} \\
    \midrule
    Road networks & \texttt{Asia\_rand}, \texttt{Italy\_rand}, \texttt{Britain\_rand} & $1.4\times10^{6}$--$1.2\times10^{7}$ & $1.6\times10^{6}$--$1.3\times10^{7}$ \\
    Power-law (Kronecker) & \texttt{kron-2e14}..\texttt{kron-2e22} & $1.6\times10^{4}$--$4.2\times10^{6}$ & $1.3\times10^{5}$--$3.4\times10^{7}$ \\
    Citation networks & \texttt{coPapersCiteseer}, \texttt{coPapersDBLP} & $4.3\times10^{5}$--$5.4\times10^{5}$ & $1.4\times10^{7}$--$1.5\times10^{7}$ \\
    DAGs & \texttt{DAG1}..\texttt{DAG3} & $5.0\times10^{2}$--$2.0\times10^{3}$ & $1.8\times10^{4}$--$4.7\times10^{4}$ \\
    Structured meshes & \texttt{torus5}, \texttt{rome99} & $3.2\times10^{1}$--$3.4\times10^{3}$ & $9.2\times10^{1}$--$7.1\times10^{3}$ \\
    \bottomrule
  \end{tabular}
  \caption{BFS evaluation datasets.}
  \label{tab:bfs-setup}
\end{table}

\subsubsection{Results and Discussion}
\label{sec:bfs-results}

We organise the BFS evaluation around the same questions as SSSP: single-thread overhead, cross-family baseline comparison, scaling behaviour, worklist sensitivity, and limitations.

\smallskip
\noindent\textbf{Single-thread overhead.}
Unlike SSSP, the LLP overhead at one thread is modest for most graph families. On road networks, LLP is actually \emph{faster} than the multi-threaded baseline even with a single thread: \texttt{Asia\_rand} finishes in \(1{,}742\)~ms (LLP-bag) versus \(2{,}461\)~ms (baseline-st), a \(0.71\times\) ratio. The same pattern holds on \texttt{Belgium\_rand} (\(0.44\times\)), \texttt{Britain\_rand} (\(0.77\times\)), and \texttt{Italy\_rand} (\(0.67\times\)). Road graphs have very low average degree (\(\approx\)1.1 edge per vertex), so the forbiddenness check touches only one or two neighbours and the compare-and-swap rarely retries. On small structured meshes and DAGs, the overhead is negligible (ratio \(\approx\)1.0). The overhead grows on denser graphs: Kronecker instances range from \(1.5\times\) (\texttt{kron-2e15}) to \(8.5\times\) (\texttt{kron-2e22}), and dense citation networks reach \(4\)--\(5\times\). As with SSSP, this reflects the cost of atomic operations applied to wide frontiers where many vertices are simultaneously active.

\smallskip
\noindent\textbf{Baseline comparison across graph families.}
Table~\ref{tab:bfs-llp-baseline} reports the best LLP configuration versus the multi-threaded baseline at 32 threads. The road network result dominates the table: PTWB reaches \(16.1\times\) on \texttt{Asia\_rand}, sustaining over 27 million vertices per second. The explanation is structural. Road graphs produce long, thin BFS frontiers---each level contains relatively few vertices spread across a large diameter. The baseline's global concurrent queue serialises access to this narrow frontier, creating contention that worsens with thread count. PTWB avoids this by giving each thread its own queue partition, so threads can drain their local frontier slabs without cross-core synchronisation.

\begin{table}[htb]
  \centering
  \resizebox{\linewidth}{!}{
  \begin{tabular}{l l r r l r l r r r r}
    \toprule
    \textbf{Family} & \textbf{Dataset} & \textbf{$|V|$ (K)} & \textbf{$|E|$ (K)} & \textbf{LLP solver} & \textbf{LLP (ms)} & \textbf{Baseline} & \textbf{Baseline (ms)} & \textbf{Speedup} & \textbf{Vertices/s (M)} & \textbf{Edges/s (M)} \\
    \midrule
    Road & \texttt{Asia\_rand} & 11950.8 & 12711.6 & PTWB & 437.39 & Base-MT & 7026.84 & $16.07\times$ & 27.32 & 29.06 \\
    Power-law & \texttt{kron-2e22} & 4194.2 & 34132.6 & PTWB & 1991.93 & Base-MT & 3084.08 & $1.55\times$ & 2.11 & 17.14 \\
    Structured & \texttt{rome99} & 3.35 & 7.08 & PTCF & 19.66 & Base-MT & 30.57 & $1.56\times$ & 0.17 & 0.36 \\
    Citation & \texttt{coPapersDBLP} & 540.49 & 15245.7 & PTWB & 403.44 & Base-MT & 1083.41 & $2.69\times$ & 1.34 & 37.79 \\
    DAG & \texttt{DAG3\_rand} & 1.00 & 18.68 & PTWB & 26.25 & Base-MT & 27.19 & $1.04\times$ & 0.04 & 0.71 \\
    Other & \texttt{r4-2e23} & 8388.6 & 33554.4 & PTWB & 1466.12 & Base-MT & 5715.60 & $3.90\times$ & 5.72 & 22.89 \\
    \bottomrule
  \end{tabular}}
  \caption{Best LLP configuration versus the multi-threaded baseline at 32 threads. Throughput is reported in millions of vertices or edges processed per second.}
  \label{tab:bfs-llp-baseline}
\end{table}

Citation networks tell a different story. \texttt{coPapersDBLP} (\(540\)k vertices, \(15.2\)M edges) yields a \(2.69\times\) speedup at 32 threads, which is respectable but far below the road-network result. The difference is that citation graphs have high average degree and community structure, producing wide frontiers where many vertices become forbidden simultaneously. In this regime, the per-thread queue partitioning still helps relative to the global queue, but the sheer volume of concurrent atomic updates limits the gain. Power-law Kronecker graphs (\texttt{kron-2e22}: \(1.55\times\)) show a similar pattern: the frontier is broad enough that queue partitioning provides only a modest edge over the baseline. DAGs stay near parity (\(1.04\times\)) because their small size exhausts available parallelism before either approach has a chance to differentiate.

\smallskip
\noindent\textbf{Scaling behaviour.}
Figure~\ref{fig:bfs-scaling} reports strong scaling on \texttt{Asia\_rand}. The LLP solver (PTWB) improves steadily from \(1{,}742\)~ms at one thread to \(437\)~ms at 32 threads, a self-speedup of \(4.0\times\). Meanwhile, the baseline actually \emph{degrades} with additional threads: it starts at \(2{,}701\)~ms at one thread and rises to \(7{,}027\)~ms at 32 threads. This anti-scaling is a direct consequence of the global queue: on a graph with diameter exceeding 20{,}000 hops, threads spend the majority of their time contending for queue access rather than doing useful traversal work. The LLP solver avoids this pathology entirely because each thread owns a private frontier partition.

On denser graphs, the scaling picture is less dramatic but still favorable. On \texttt{coPapersDBLP}, the baseline is \(4\times\) faster than LLP at one thread (\(833\)~ms vs \(3{,}286\)~ms), but by eight threads LLP overtakes (\(742\)~ms vs \(1{,}069\)~ms, \(1.44\times\)) and widens its lead to \(2.69\times\) at 32 threads. On \texttt{kron-2e22}, the crossover comes later: LLP trails through eight threads and only overtakes at 16 threads (\(1.34\times\)), reaching \(1.55\times\) at 32.

\begin{figure}[htb]
  \centering
  \includegraphics[width=0.48\textwidth]{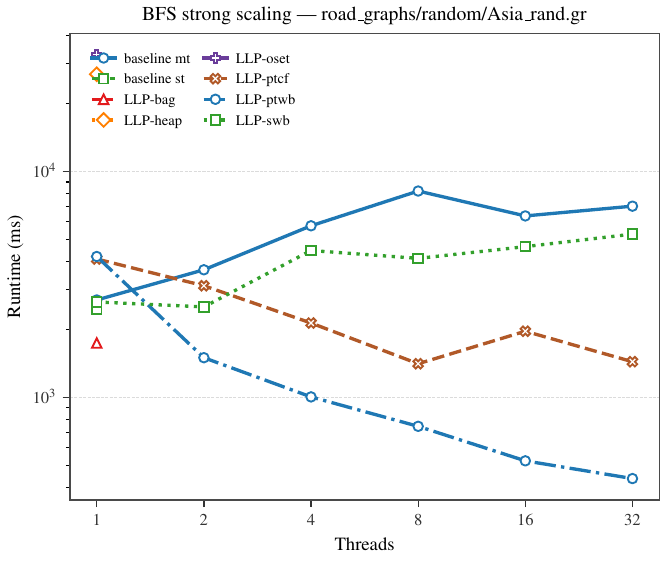}\hfill
  \includegraphics[width=0.48\textwidth]{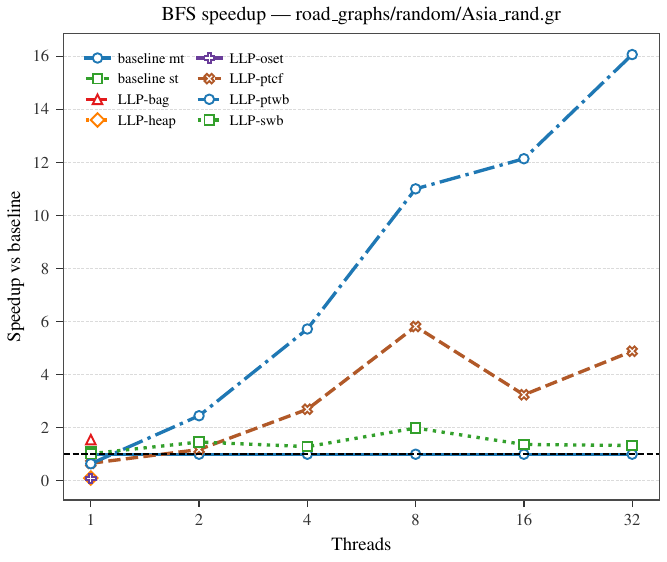}
  \caption{BFS runtime (left) and speedup relative to the baseline (right) on \texttt{Asia\_rand}.}
  \label{fig:bfs-scaling}
\end{figure}

\smallskip
\noindent\textbf{Worklist policy.}
Figure~\ref{fig:bfs-worklists} shows the worklist comparison on \texttt{Asia\_rand} at 32 threads. PTWB finishes in \(437\)~ms, while the next-best LLP policy (PTCF) requires \(1{,}439\)~ms and SWB needs \(5{,}285\)~ms. The gap between PTWB and SWB is a \(12\times\) difference that reflects the same mechanism as in SSSP: per-thread partitioning keeps each thread's working set cache-local and avoids global-queue contention. However, unlike SSSP, the scheduler advantage in BFS does not come from a priority signal (there are no edge weights to prioritise by). Instead, the benefit is purely spatial: PTWB assigns each thread a contiguous range of neighbour indices, so successive cache lines are accessed by the same core rather than being bounced between sockets.

\begin{figure}[htb]
  \centering
  \includegraphics[width=0.62\linewidth]{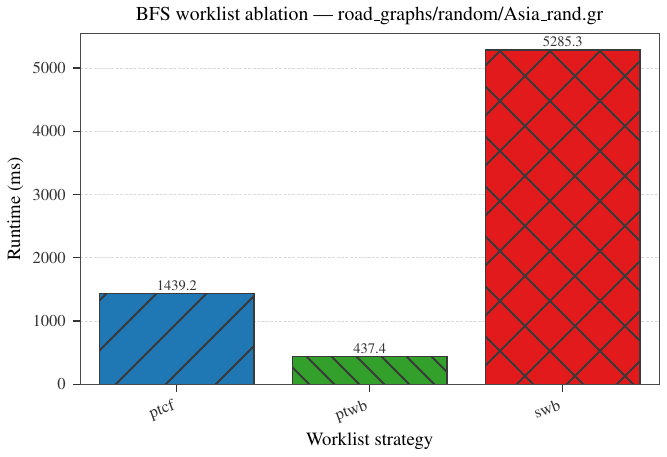}
  \caption{Worklist diagnostics for BFS on \texttt{Asia\_rand} with 32 threads.}
  \label{fig:bfs-worklists}
\end{figure}

\smallskip
\noindent\textbf{Comparison with SSSP.}
The BFS results reinforce and refine the patterns observed in SSSP. Both algorithms show the same structural dependency: LLP gains are largest on graphs with narrow, elongated frontiers (road networks) and smallest on graphs with wide, dense frontiers (power-law, citation). The key difference is the single-thread overhead profile. In SSSP, LLP is consistently slower at one thread (\(\approx\)10\(\times\)) because the forbiddenness check involves reading all neighbour distances and comparing against edge weights. In BFS, the check is cheaper (just compare integer levels), so the single-thread overhead is much lower and on road graphs LLP is actually faster even sequentially. This means the crossover to LLP dominance happens earlier in BFS than in SSSP: on \texttt{Asia\_rand}, LLP leads from the very first thread, whereas SSSP required four or more threads to overtake.

\smallskip
\noindent\textbf{Summary.}
BFS confirms that the LLP runtime transfers effectively to unweighted traversal. Road networks are the standout workload, with \(16\times\) gains at 32 threads driven by the elimination of global-queue contention on narrow frontiers. Citation and power-law graphs benefit more modestly (\(1.5\)--\(2.7\times\)), and tiny graphs show parity. Worklist policy remains important, but the mechanism shifts from priority-based scheduling (as in SSSP) to spatial partitioning---a distinction that matters for practitioners choosing scheduler configurations.

\subsection{Stable Marriage}
\label{sec:sm}

\subsubsection{Formulation}
The classical Stable Marriage problem matches two equally sized populations (traditionally called \emph{men} and \emph{women}) such that no blocking pair exists. The standard sequential algorithm is Gale--Shapley: men propose in rounds, women hold the best offer received so far, and the process terminates when every man is either matched or has exhausted his list. Parallelising Gale--Shapley is non-trivial because a proposal that displaces one man can cascade into further proposals, creating serial dependencies.

The LLP formulation sidesteps round-based coordination. Each man \(m\) tracks a proposal index into his preference list; the lattice is ordered pointwise on these indices, so proposals only move forward and the state is monotone. A man is \emph{forbidden} if the woman he is currently targeting prefers another man who also currently targets her---in other words, there exists a blocking pair involving \(m\). The advance step increments \(m\)'s proposal index toward the next woman on his list and enqueues any men who may have become forbidden as a result (those currently targeting the woman \(m\) just proposed to). Because indices never decrease, the global state converges monotonically to the unique man-optimal stable matching, which is the same fixed point that Gale--Shapley reaches.

The key difference from Gale--Shapley is scheduling granularity. Gale--Shapley processes one proposal per round (or one per man per round in the parallel variant); the LLP solver can advance any forbidden man at any time, so multiple independent proposals proceed concurrently without waiting for a round boundary. The cost is the same as in the graph algorithms: each proposal requires an atomic compare-and-swap on the proposal index, and the forbiddenness check must read the current matches of all competitors for the target woman.

\subsubsection{Implementation}
LLP-FW keeps the immutable preference lists, per-participant proposal indices, and the worklist scheduler separate. Listing~\ref{lst:sm-struct} sketches the data representation, while Listing~\ref{lst:sm-impl} outlines the \texttt{LatticeLinearProblem} implementation used by the runtime. Two aspects of the implementation deserve attention. First, the forbiddenness check for man \(m\) iterates over all men to find competitors currently targeting the same woman---this is an \(O(n)\) scan per predicate evaluation, which becomes the dominant cost as \(n\) grows. Second, a successful proposal enqueues the affected men (those who may now be displaced) back onto the worklist, creating a cascade pattern similar to the neighbour-push in SSSP. The PTCF scheduler groups these cascading re-evaluations into per-thread chunks, amortising the overhead of worklist management.

\begin{figure}[htb]
  \begin{minipage}{\linewidth}
  \begin{lstlisting}[language=rust,caption=Stable Marriage data representation,label={lst:sm-struct},basicstyle=\ttfamily\scriptsize]
pub struct LlpStableMarriage {
  mprefs: Vec<Vec<usize>>, // Men's preferences
  wprefs: Vec<Vec<usize>>, // Women's preferences
}

impl LlpStableMarriage {
  // Returns true if person i prefers their current match over person j
  fn prefers_current_match(&self,
    global_state: &GlobalState<usize>,
    i: usize,
    j: usize) -> bool;
   
  // Compare-and-swap the proposal index,
  // returning the actual index if CAS fails
  fn compare_and_swap(&self,
    global_state: &GlobalState<usize>,
    index: usize,
    new_index: usize) -> Result<(), usize>;
       
  // Returns the set of men who might become forbidden
  // due to a new proposal
  fn get_affected_men(&self,
    proposer: usize,
    woman: usize) -> Vec<usize>;
}
  \end{lstlisting}
  \end{minipage}
\end{figure}
   
\begin{figure}[htb]
\begin{minipage}{\linewidth}
\begin{lstlisting}[language=rust,caption=Stable Marriage as an LLP,label={lst:sm-impl},basicstyle=\ttfamily\scriptsize]
impl LatticeLinearProblem<usize> for LlpStableMarriage {
  fn is_forbidden(&self, global_state: &GlobalState<usize>, index: usize) -> bool {
    let current_woman = self.mprefs[index][global_state.read(index)];
     
    // Check if any man has a better proposal from current_woman
    (0..self.mprefs.len()).any(|i| 
      i != index && 
      self.mprefs[i][global_state.read(i)] == current_woman &&
      self.prefers_current_match(global_state, i, index)
    )
  }
   
  fn advance<W: Worklist>(
    &self,
    global_state: &GlobalState<usize>,
    index: usize,
    worklist: &W,
  ) -> bool {
    let mut current_idx = global_state.read(index);
     
    loop {
      let new_idx = current_idx + 1;
      let new_woman = self.mprefs[index][new_idx];
   
      match self.compare_and_swap(global_state, index, new_idx) {
        Ok(_) => {
          worklist.push_all(self.get_affected_men(index, new_woman));
          return true;
        }
        Err(actual) => current_idx = actual,
      }
    }
  }
}
\end{lstlisting}
\end{minipage}
\end{figure}

\subsubsection{Setup}
We evaluate five balanced Stable Marriage instances with randomly generated total-order preferences, ranging from 10 to 10\,000 participants per side. Table~\ref{tab:sm-setup} summarises the workloads. The baseline is a parallel Gale--Shapley implementation that batches proposals per round and uses a shared concurrent queue; a sequential (non-parallel) Gale--Shapley is also included for single-thread reference. The LLP solver uses the per-thread chunk FIFO (PTCF) worklist. Multi-thread data covers 1--32 threads across all five instances.

\begin{table}[htb]
  \centering
  \small
  \begin{tabular}{lrrr}
    \toprule
    \textbf{Workload} & \textbf{Men} & \textbf{Women} & \textbf{Avg prefs} \\
    \midrule
    Sm 10    & 10    & 10    & 10    \\
    Sm 100   & 100   & 100   & 100   \\
    Sm 1000  & 1000  & 1000  & 1000  \\
    Sm 5000  & 5000  & 5000  & 5000  \\
    Sm 10000 & 10000 & 10000 & 10000 \\
    \bottomrule
  \end{tabular}
  \caption{Stable Marriage workloads used in the evaluation. Preference lists are balanced and dense for every instance.}
  \label{tab:sm-setup}
\end{table}

\subsubsection{Results and Discussion}
\label{sec:sm-results}

Stable Marriage differs from the graph algorithms in two important ways: the problem size is measured in participants rather than vertices and edges, and the forbiddenness check is \(O(n)\) per evaluation rather than \(O(\text{degree})\). These differences change the overhead and scaling profiles, and the results below reflect that.

\smallskip
\noindent\textbf{Single-thread overhead.}
At one thread, the LLP solver PTCF is comparable to sequential Gale-Shapley for medium-sized instances. For \textsc{Sm 100}, PTCF takes \(0.292\)~ms, whereas the sequential Gale-Shapley takes \(0.032\)~ms, an overhead of roughly \(9\times\). However, for \textsc{Sm 1000}, PTCF takes \(6.262\)~ms against the sequential Gale-Shapley's \(6.921\)~ms, which is effectively parity (PTCF is actually \(10\%\) faster). For larger instances, the overhead reappears: on \textsc{Sm 10000} PTCF takes \(2{,}092\)~ms against the sequential baseline's \(960\)~ms, a \(2.2\times\) overhead. This suggests that while per-operation atomic costs are amortised at intermediate sizes, the \(O(n)\) forbiddenness scan dominates at the largest scale for single-threaded execution.

\smallskip
\noindent\textbf{Baseline comparison: speedup grows with instance size.}
Table~\ref{tab:sm-llp-baseline} compares the 32-thread LLP solver against the parallel Gale--Shapley baseline. The headline result is the growth of the speedup with instance size: from \(2.27\times\) on \textsc{Sm 100} to \(245.88\times\) on \textsc{Sm 1000}, then \(148.03\times\) on \textsc{Sm 5000} and \(109.43\times\) on \textsc{Sm 10000}.

\begin{table}[htb]
  \centering
  \small
  \begin{tabular}{lrrr}
    \toprule
    \textbf{Workload} & \textbf{LLP-PTCF (ms)} & \textbf{Baseline (ms)} & \textbf{Speedup} \\
    \midrule
    Sm 10    & 1.202  & 0.927     & $0.77\times$ \\
    Sm 100   & 1.231  & 2.793     & $2.27\times$ \\
    Sm 1000  & 1.791  & 440.380   & $245.88\times$ \\
    Sm 5000  & 17.564 & 2{,}599.963 & $148.03\times$ \\
    Sm 10000 & 74.960 & 8{,}202.645 & $109.43\times$ \\
    \bottomrule
  \end{tabular}
  \caption{32-thread LLP-PTCF versus the parallel Gale--Shapley baseline across all instance sizes.}
  \label{tab:sm-llp-baseline}
\end{table}

The peak speedup occurs on \textsc{Sm 1000} because the parallel Gale--Shapley baseline anti-scales most severely at this size: at 32 threads it runs in \(440\)~ms, nearly \(2\times\) slower than its own single-thread variant (\(238\)~ms). The round-based structure forces all proposals in a round to finish before the next can begin, and each displaced man must wait for the next round to re-propose. In the LLP solver, a displaced man is immediately re-enqueued and can re-propose in the same scheduling epoch.

On the larger instances (\textsc{Sm 5000} and \textsc{Sm 10000}), the speedup is somewhat lower because the \(O(n)\) forbiddenness check increases the absolute cost of each LLP advance. Nevertheless, the baseline's anti-scaling means that the LLP advantage remains over two orders of magnitude on \textsc{Sm 5000} and over one hundred fold on \textsc{Sm 10000}.

\smallskip

\noindent\textbf{Scaling behaviour on \textsc{Sm 10000}.}

Figure~\ref{fig:sm-scaling} shows the scaling performance on the largest instance. The LLP solver (PTCF) scales from \(2{,}092\)~ms at one thread to \(74.96\)~ms at 32 threads, a self-speedup of \(27.9\times\)---near-linear scaling. In contrast, the parallel baseline \emph{anti-scales}: it goes from \(5{,}728\)~ms at one thread to \(8{,}203\)~ms at 32 threads, a \(1.4\times\) slowdown. This divergence means the LLP speedup over the baseline grows from \(2.7\times\) at one thread to \(109\times\) at 32 threads.

\begin{figure}[htb]
  \centering
  \includegraphics[width=0.48\textwidth]{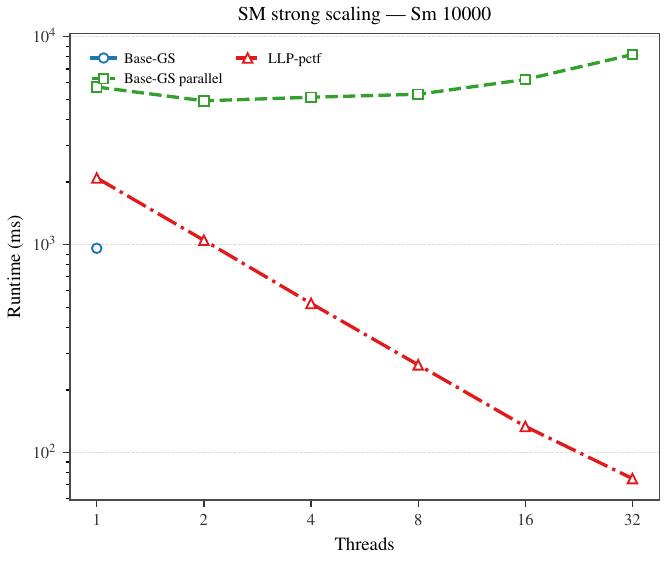}
  \hfill
  \includegraphics[width=0.48\textwidth]{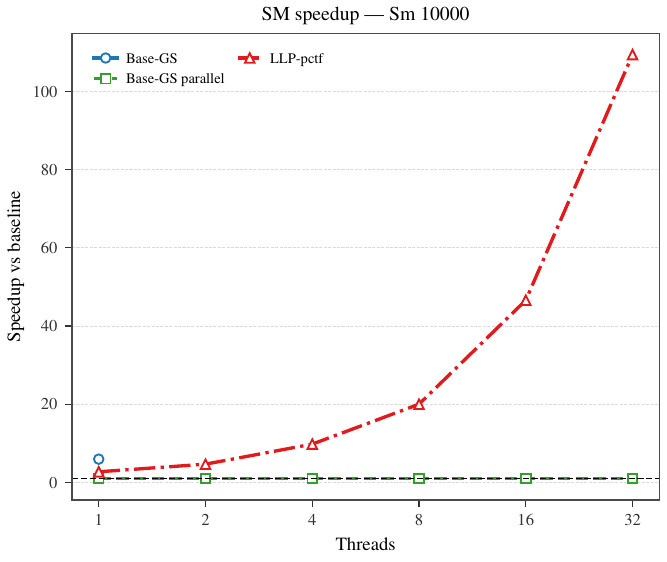}
  \caption{Stable Marriage runtime (left) and speedup over parallel Gale--Shapley (right) on \textsc{Sm 10000}.}
  \label{fig:sm-scaling}
\end{figure}

\smallskip
\noindent\textbf{Scheduler policy: cyclic traversal is catastrophic at scale.}
Table~\ref{tab:sm-worklists} compares single-thread runtimes for PTCF and the cyclic traversal across the smaller instance sizes (the cyclic solver does not terminate within a reasonable time on \textsc{Sm 5000} and above). On \textsc{Sm 10} (10 participants) the cyclic variant is actually the fastest option (\(0.004\)~ms). On \textsc{Sm 100} the balance shifts: a full scan now touches 100 participants, most of whom are not forbidden, and PTCF's targeted re-evaluation is over \(17\times\) faster (\(0.292\)~ms vs \(5.250\)~ms). On \textsc{Sm 1000} the gap reaches three orders of magnitude (\(6.262\)~ms vs \(11{,}916\)~ms). The lesson is that targeted worklist management is essential for Stable Marriage at scale: a naive full-scan approach turns an \(O(n^2)\) problem into something closer to \(O(n^3)\) or worse in practice.

\begin{table}[htb]
  \centering
  \small
  \begin{tabular}{lrr}
    \toprule
    \textbf{Workload} & \textbf{LLP-PTCF (ms)} & \textbf{LLP-T-Cyclic (ms)} \\
    \midrule
    Sm 10   & 0.145  & 0.004 \\
    Sm 100  & 0.292  & 5.250 \\
    Sm 1000 & 6.262  & 11{,}916 \\
    \bottomrule
  \end{tabular}
  \caption{Single-thread (1-thread) performance of two LLP schedulers. The cyclic traversal becomes orders of magnitude slower once preference lists exceed a few dozen entries.}
  \label{tab:sm-worklists}
\end{table}

\smallskip
\noindent\textbf{Comparison with graph algorithms.}
Compared to the SSSP and BFS results, the Stable Marriage problem demonstrates a different scaling profile for the LLP abstraction. For the graph problems, the benefits arise from avoiding global barriers during frontier expansion. For Stable Marriage, the benefits arise from avoiding global rounds during cascading re-proposals. The LLP solver achieves near-linear self-speedup (\(27.9\times\) at 32 threads on \textsc{Sm 10000}), whereas the parallel Gale--Shapley baseline exhibits consistent anti-scaling. The \(O(n)\) cost of the forbiddenness check is one limitation: an optimised implementation that maintains per-woman competitor sets could reduce this to \(O(1)\) amortised cost, likely extending the LLP advantage further.

\smallskip
\noindent\textbf{Summary.}
Stable Marriage is the strongest result in this paper in terms of baseline-relative speedup. The \(246\times\) gain on \textsc{Sm 1000} at 32 threads reflects both the LLP solver's ability to overlap cascading proposals and the parallel baseline's inability to do so. On the largest instance (\textsc{Sm 10000}), LLP-PTCF achieves \(109\times\) over the baseline at 32 threads with near-linear self-speedup (\(27.9\times\)), while the baseline anti-scales. Scheduler choice remains critical: the cyclic traversal degrades by three orders of magnitude on 1\,000-participant instances, while PTCF's targeted re-evaluation keeps runtime proportional to actual work. These results demonstrate that the LLP abstraction applies naturally to combinatorial matching problems where cascading updates create serial dependencies that round-based parallelism cannot exploit.

After graph and matching workloads, we now apply the same analysis template to four additional domains (scheduling, reductions, closure, and dynamic programming) to test how robustly the LLP runtime transfers beyond traversal-heavy graphs.

\subsection{Job Scheduling}
\label{sec:job-scheduling}

Job scheduling with precedence constraints arises in build systems, manufacturing cells, and distributed execution frameworks. Given a Directed Acyclic Graph (DAG) whose vertices represent jobs and edges encode prerequisites, the objective is to compute each job's earliest completion time. The standard sequential approach is a topological sort (Kahn's algorithm): process jobs in dependency order, computing each completion time from the maximum of its predecessors. The parallel variant processes one topological \emph{level} per round, advancing all jobs in a level simultaneously before starting the next. This level-synchronous structure is simple but conservative: it waits for every job in a level to finish before releasing successors, even when some successors could start earlier because their actual predecessors completed in a prior level.

The LLP formulation removes this barrier. A job becomes forbidden as soon as all its predecessors have committed their completion times, regardless of what other jobs at the same topological depth are doing. The solver advances forbidden jobs immediately and enqueues their successors, so independent chains in the DAG proceed concurrently without waiting for the widest level to drain.

\subsubsection{Formulation}
\label{sec:job-formulation}

\noindent\textbf{Problem statement.} Let \(G = (V,E)\) be a DAG with \(|V| = n\). Each job \(j \in V\) has processing time \(t_j \in \mathbb{R}_{\ge 0}\) and prerequisite set \(P(j) = \{ i \mid (i,j) \in E \}\). We seek the vector of earliest completion times \(\C[j]\) that satisfies all precedence constraints, i.e.,
\[
\C[j] \ge \max_{p \in P(j)} \C[p] + t_j \quad \text{for every } j \in V,
\]
where the maximum over an empty predecessor set is zero.

\noindent\textbf{Lattice and forbidden predicate.} The global state \(\C \in \mathbb{R}_{\ge 0}^n\) forms a lattice under component-wise order. We initialise \(\C[j] = t_j\), reflecting that jobs cannot finish before their own execution time. A job becomes \emph{ready} once all predecessors are fixed:
\[
\textsc{ready}(j) \equiv \forall p \in P(j)\; \textsc{fixed}(p).
\]
The lattice-linear predicate \(B(\C)\) checks whether every ready job already respects its precedence bound. Violations manifest as forbidden indices
\[
\forbidden(\C, j, B) \equiv \textsc{ready}(j) \wedge \C[j] < \max_{p \in P(j)} \C[p] + t_j.
\]
Advancing a forbidden job raises \(\C[j]\) to the tight bound while preserving monotonicity. Because successors only become ready after their last predecessor is fixed, updates propagate along the DAG without reintroducing stale states.

\subsubsection{Implementation}
\label{sec:job-implementation}

The concrete solver (\texttt{LlpJobScheduling} in \texttt{src/algorithms/job.rs}) maintains two auxiliary arrays alongside the solution vector: \texttt{max\_from\_parents} caches the largest completion published by any predecessor, and \texttt{remaining\_prereqs} tracks how many dependencies are still outstanding. Sources enqueue themselves during \texttt{initial\_states\_to\_process} with a priority equal to their predicted completion time. The combination ensures that threads only pop jobs whose parents are finished and that they receive a priority hint for downstream scheduling heuristics.

Listing~\ref{lst:job-advance} sketches the \texttt{advance} routine. It (i) raises the job's completion time through a compare-and-swap loop, (ii) fixes the job exactly once via the lock-free fixed-state bit vector (see Section~\ref{sec:cm-memoization}), and (iii) publishes readiness to each successor using a \texttt{fetch\_max} on the parent-maximum cache followed by a \texttt{fetch\_sub} on the prerequisite counter---when the counter reaches zero, the child is unlocked. Ready children are reinserted into the worklist with a priority equal to their own earliest completion estimate, biasing the solver toward keeping the critical path warm. The memoization through fixed states is especially effective here: once a job's completion time is committed, it is never revisited, so the fixed bit eliminates redundant forbiddenness checks on settled jobs entirely.

\begin{figure}[htb]
\begin{minipage}{\linewidth}
\begin{lstlisting}[language=rust,caption=Core of the LLP job scheduling advance step.,label={lst:job-advance},basicstyle=\ttfamily\scriptsize]
impl LatticeLinearProblem<usize, usize, JobState> for LlpJobScheduling {
    fn advance<W: Worklist>(
        &self,
        global_state: &GlobalState<usize, JobState>,
        index: usize,
        worklist: &W,
    ) -> bool {
        let job = &self.jobs[index];
        let parent_max = global_state
            .additional_state
            .max_from_parents[index]
            .load(Acquire);
        let desired_completion = parent_max.saturating_add(job.time);

        let mut stored = global_state.solution_vector[index].load(Acquire);
        let mut updated = false;
        while desired_completion > stored {
            match global_state.solution_vector[index].compare_exchange_weak(
                stored,
                desired_completion,
                Release,
                Acquire,
            ) {
                Ok(_) => {
                    stored = desired_completion;
                    updated = true;
                    break;
                }
                Err(actual) => stored = actual,
            }
        }

        let was_fixed = global_state.fixed_vector.is_fixed(index);
        if !was_fixed {
            global_state.fixed_vector.set_fixed(index);
            let completion = stored;
            let mut ready = Vec::new();
            for &child in &job.postreq {
                global_state.additional_state.max_from_parents[child]
                    .fetch_max(completion, AcqRel);
                if global_state.additional_state.remaining_prereqs[child]
                    .fetch_sub(1, AcqRel)
                    == 1
                {
                    let priority = global_state
                        .additional_state
                        .max_from_parents[child]
                        .load(Acquire)
                        .saturating_add(self.jobs[child].time);
                    ready.push((child, priority));
                }
            }
            if !ready.is_empty() {
                worklist.push_all(ready.into_iter());
            }
        }

        updated
    }
}
\end{lstlisting}
\end{minipage}
\end{figure}

\subsubsection{Setup}
\label{sec:job-setup}

We generate synthetic DAGs ranging from ten to ten thousand jobs using the script in \texttt{benches/data/job\_scheduling\_gen.py}. Each instance samples job durations uniformly in \([1,80]\) and adds a directed edge from every earlier job with probability \(0.2\), yielding moderately wide but acyclic dependency graphs. We compare against two baselines: a sequential Kahn-style topological pass (\texttt{baseline - topo sort}) that computes completion times in strict dependency order, and its level-synchronous parallel variant (\texttt{baseline - parallel topo levels}) that processes one topological frontier per round using a shared concurrent queue. All solvers run on the platform described in Section~\ref{sec:cm-eval} with thread counts from 1 to 32; we report per-dataset medians over repeated Criterion samples.

\begin{table}[htb]
    \centering
    \small
    \begin{tabular}{lr}
        \toprule
        \textbf{Dataset} & \textbf{Jobs} \\
        \midrule
        \texttt{Job 10} & 10 \\
        \texttt{Job 100} & 100 \\
        \texttt{Job 300} & 300 \\
        \texttt{Job 500} & 500 \\
        \texttt{Job 1000} & 1\,000 \\
        \texttt{Job 10000} & 10\,000 \\
        \bottomrule
    \end{tabular}
    \caption{Synthetic job scheduling benchmarks. All DAGs draw durations uniformly from \([1,80]\) and include each potential prerequisite with probability \(0.2\).}
    \label{tab:job-setup}
\end{table}

\subsubsection{Results and Discussion}
\label{sec:job-eval}

Job scheduling offers a qualitatively different performance story from the graph algorithms. Because the memoization optimisation prevents any settled job from being re-evaluated, the LLP solver avoids the repeated scans that dominate the topological baselines on deep DAGs. The consequence is that LLP is already dramatically faster at \emph{one thread} on large instances, and multi-threading widens the gap further.

\smallskip
\noindent\textbf{Single-thread behaviour: parity on small DAGs, orders-of-magnitude win on large ones.}
On instances with up to 500 jobs, the best single-thread LLP solver (typically \emph{oset} or \emph{bag}) and the sequential topological sort finish within a few percent of each other---both around \(13\)--\(15\)~ms. The DAGs are small enough that the total work is comparable regardless of traversal strategy. The picture changes sharply on \texttt{Job 10000}: the sequential topological baseline requires \(261.7\)~ms, while the best LLP variant (\emph{oset}) finishes in \(14.7\)~ms---a \(17.8\times\) speedup with a single thread. The explanation lies in how each approach handles the deep dependency chains that the random-DAG generator produces at this scale. The topological sort recomputes ready sets from scratch on every pass; the LLP solver computes each job's completion time once, marks it as fixed, and never revisits it. Successors read the cached parent maximum directly, so the total number of operations scales with the number of edges rather than with repeated vertex scans.

\smallskip
\noindent\textbf{Scaling on \texttt{Job 10000}.}
Figure~\ref{fig:job-scaling} reports runtimes and speedups for the largest workload. The level-synchronous baseline shows essentially flat performance across thread counts: it starts at \(372\)~ms with one thread and fluctuates between \(411\)~ms and \(463\)~ms through 32 threads. The reason is that each round must drain completely before the next topological level is released, and on a deep DAG with narrow bottleneck levels, most threads sit idle waiting for the few jobs on the critical path to finish. The LLP solver faces no such barrier. At one thread it already runs in \(14.7\)~ms; adding threads increases it modestly to \(34.3\)~ms at 32 threads (the slight slowdown reflects coordination overhead exceeding the available parallelism on this particular DAG). The baseline-relative speedup is therefore dominated by the single-thread gap: \(13.5\times\) at 32 threads, and the absolute best point is at one thread (\(17.8\times\)).

This is an unusual scaling profile compared with SSSP and BFS, where LLP started slower and overtook the baseline as threads grew. Here, the algorithmic advantage (memoization + demand-driven successor propagation) already dominates at one thread, and additional threads provide only marginal further improvement because the 10\,000-job DAG does not expose enough independent work to keep 32 cores busy once the critical path is the bottleneck.

\begin{figure}[htb]
    \centering
    \includegraphics[width=0.48\textwidth]{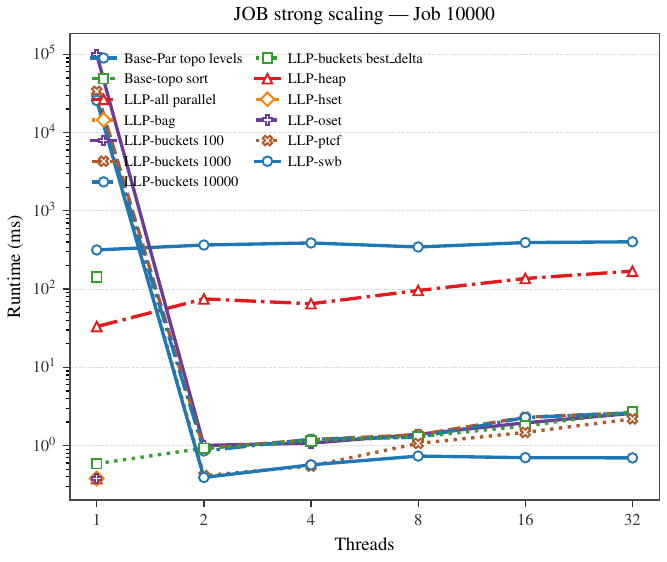}\hfill
    \includegraphics[width=0.48\textwidth]{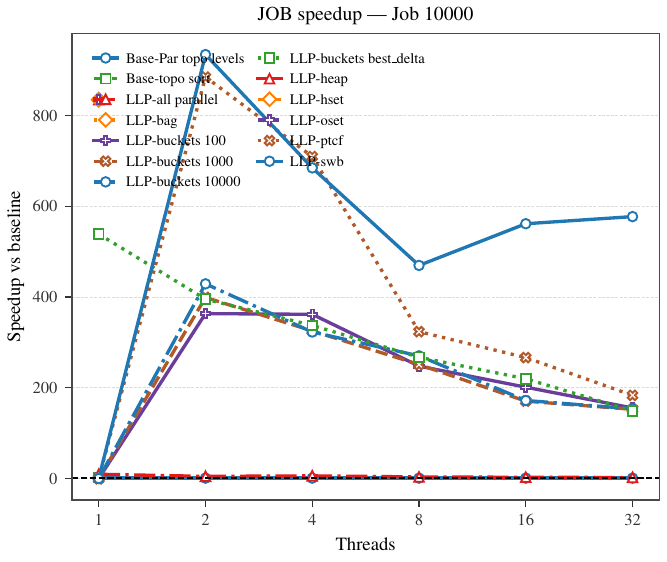}
    \caption{Job scheduling strong scaling (left) and speedup versus the parallel topological baseline (right) on \texttt{Job 10000}.}
    \label{fig:job-scaling}
\end{figure}

On intermediate DAG sizes the parallelism balance is different. On \texttt{Job 1000} at four threads, LLP finishes in \(16.5\)~ms versus \(29.4\)~ms for the baseline (\(1.78\times\)). At 32 threads the gap narrows to \(1.40\times\) (\(31.6\)~ms vs \(44.2\)~ms), consistent with the smaller DAG exhausting available parallel slack sooner.

\smallskip
\noindent\textbf{Worklist sensitivity.}
Figure~\ref{fig:job-worklists} compares scheduler policies on \texttt{Job 10000} at 32 threads. Unlike the graph algorithms, where PTWB dominated by a wide margin, the job-scheduling worklists cluster tightly: PTCF at \(34.3\)~ms, SWB at \(35.9\)~ms, all-parallel at \(37.3\)~ms, and several bucketed variants between \(38\)~ms and \(39\)~ms. The gap between the best and worst LLP policy is only about \(15\%\), compared with the order-of-magnitude gaps observed in SSSP and BFS.

The reason is that the memoization optimisation changes the nature of the frontier. In SSSP, the same vertex can be relaxed multiple times, so prioritising recently improved vertices (PTWB's recency bias) eliminates a large volume of redundant work. In job scheduling, each job is advanced exactly once and then fixed; there is no redundant work to avoid. The worklist's role is reduced to ordering the sequence in which ready jobs are discovered, which matters less when every job will be processed exactly once regardless of order.

\begin{figure}[htb]
    \centering
    \includegraphics[width=0.48\textwidth]{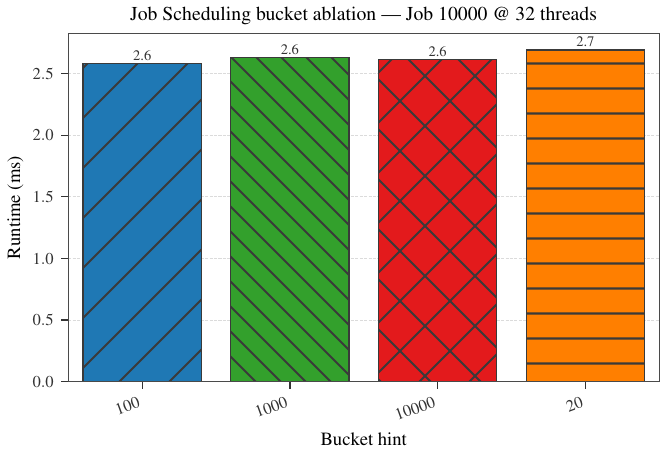}\hfill
    \includegraphics[width=0.48\textwidth]{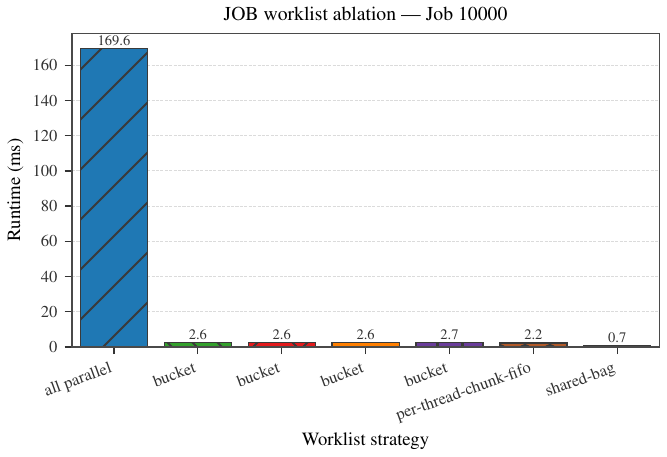}
    \caption{Scheduler ablation on \texttt{Job 10000} at 32 threads. Left: bucket sizes versus runtime; right: active frontier size over time.}
    \label{fig:job-worklists}
\end{figure}

\smallskip
\noindent\textbf{Comparison with graph algorithms and summary.}
Job scheduling occupies a distinct niche in the LLP portfolio. In SSSP and BFS, the LLP overhead at low thread counts is the main cost, and gains come from better parallel scaling. In job scheduling, the gains come primarily from the \emph{algorithmic} advantage of memoization and demand-driven propagation, which eliminates redundant work even sequentially. Multi-threading adds a modest further benefit on DAGs with sufficient width, but the critical-path depth limits the available parallelism on our synthetic instances.

For a practitioner, this means LLP is especially attractive for scheduling workloads where the dependency graph is deep and the baseline involves repeated global scans. The worklist choice matters less than in graph traversal---any reasonable policy delivers similar performance---so the implementation can default to a simple scheduler without extensive tuning.

\subsection{Parallel Reduction}
\label{sec:reduce}

\subsubsection{Formulation}
Parallel reduction collapses an input vector \(x \in \mathbb{R}^n\) into a single scalar by repeatedly applying an associative operator (we use 64-bit integer addition). The standard parallel approach builds a binary tree of partial sums and combines them bottom-up; each level of the tree depends only on the level below, so all combines at the same depth can proceed in parallel.

In the LLP view, each node of the reduction tree is a lattice coordinate. A node becomes \emph{forbidden} when both of its children have published their partial sums but the node itself has not yet combined them. Advancing the state performs the combine via an atomic update and exposes the parent as a candidate for further reduction. The lattice order is pointwise: once a partial sum is published it never reverts. This formulation maps the reduction tree directly onto the LLP runtime, but it also introduces a cost that the binary-tree baseline avoids: every combine step requires an atomic read-modify-write rather than a plain memory store.

We include parallel reduction deliberately as a stress test for the LLP abstraction. Reduction is the most regular, most memory-bandwidth-sensitive workload in our benchmark suite, with no irregular frontier structure for the LLP scheduler to exploit. If the framework performs well here, it would suggest broad applicability; if it does not, the gap will reveal the cost of generality on workloads that do not benefit from demand-driven scheduling.

\subsubsection{Implementation}
The LLP implementation assigns each tile of the reduction tree to a work-item containing a pointer to its accumulator and a dependency counter tracking outstanding children. Tiles publish their partial sums through \texttt{AtomicU64} updates, at which point the dependency counter of their parent is decremented. When the counter reaches zero the parent becomes forbidden and is pushed onto the active worklist. PTWB and SWB favour locality by letting a thread drain the subtree rooted at its current tile; AllPar behaves more like a flat tree and maximises available parallelism at the cost of extra cross-core traffic. The baseline is a multi-threaded binary-tree reduction optimised for cache reuse, using plain (non-atomic) stores within each thread's local subtree and a single synchronisation point at the root.

\subsubsection{Setup}
We generate uniformly distributed 64-bit inputs and evaluate four power-of-two sizes that stress different cache hierarchy levels. Table~\ref{tab:reduce-setup} summarises the dataset suite. All runs use the platform described in Section~\ref{sec:cm-eval} with thread counts from 1 to 32; reported values are medians over repeated Criterion samples.

\begin{table}[htb]
    \centering
    \small
    \begin{tabular}{lcl}
        \toprule
        \textbf{Dataset} & \textbf{Elements} & \textbf{Notes} \\
        \midrule
        \texttt{Reduce 4096} & $4{,}096$ & Fits in private L1 caches \\
        \texttt{Reduce 65536} & $65{,}536$ & Resident in shared L2 \\
        \texttt{Reduce 1048576} & $1{,}048{,}576$ & Spans the LLC \\
        \texttt{Reduce 16777216} & $16{,}777{,}216$ & Bandwidth-bound, spills to DRAM \\
        \bottomrule
    \end{tabular}
    \caption{Parallel reduction benchmarks. All runs reduce 64-bit integers with associative addition and reuse the same random seed across solvers.}
    \label{tab:reduce-setup}
\end{table}

\subsubsection{Results and Discussion}
\label{sec:reduce-results}

Parallel reduction is the weakest workload for LLP in this paper, and the results illustrate where the abstraction's overhead outweighs its benefits. We present them honestly because the failure mode is as informative as the successes.

\smallskip
\noindent\textbf{Single-thread overhead.}
At one thread, the LLP solver pays a substantial penalty on every combine step. On \texttt{Reduce 4096} the best LLP variant (\emph{bag}) takes \(0.053\)~ms versus \(0.006\)~ms for the sequential baseline---roughly \(9\times\) slower. On \texttt{Reduce 1M} the gap is similar: \(16.2\)~ms versus \(1.72\)~ms (\(9.4\times\)). On \texttt{Reduce 16M} the ratio shrinks to \(2.7\times\) (\(285\)~ms vs \(105\)~ms) because the combine loop now saturates memory bandwidth and the atomic overhead becomes a smaller fraction of total time. Unlike SSSP or job scheduling, where the LLP predicate check does useful filtering work, the reduction predicate is trivial (``are both children ready?''), so the compare-and-swap cost is pure overhead with no compensating benefit.

\smallskip
\noindent\textbf{Baseline comparison at 32 threads.}
Table~\ref{tab:reduce-llp-baseline} reports the best LLP configuration versus the parallel baseline at 32 threads across all four input sizes. On the two smallest inputs (\(4{,}096\) and \(65{,}536\) elements), both approaches finish in the same ballpark because the total work fits in cache and runtime is dominated by thread-infrastructure overhead rather than actual combines. On \texttt{Reduce 1M} the gap opens: the baseline reaches \(26.6\)~ms while PTWB requires \(45.4\)~ms (\(0.59\times\)). On \texttt{Reduce 16M} the deficit becomes decisive: \(96.5\)~ms versus \(417.3\)~ms (\(0.23\times\)). The throughput column tells the same story---the baseline sustains \(174\) million elements per second on the largest input, while LLP manages only \(40\) million.

\begin{table}[htb]
    \centering
    \resizebox{\linewidth}{!}{%
    \begin{tabular}{lcccccc}
        \toprule
        \textbf{Dataset} & \textbf{Elements} & \textbf{Baseline (ms)} & \textbf{Baseline (M elems/s)} & \textbf{Best LLP} & \textbf{LLP (ms)} & \textbf{LLP (M elems/s)} \\
        \midrule
        \texttt{Reduce 4096}  & $4{,}096$       & 18.22 & 0.22 & AllPar & 18.22 & 0.22 \\
        \texttt{Reduce 65536} & $65{,}536$      & 20.32 & 3.23 & PTWB   & 18.84 & 3.48 \\
        \texttt{Reduce 1048576}  & $1{,}048{,}576$ & 26.58 & 39.45 & PTWB   & 45.35 & 23.12 \\
        \texttt{Reduce 16777216} & $16{,}777{,}216$& 96.50 & 173.86 & PTWB  & 417.29 & 40.21 \\
        \bottomrule
    \end{tabular}}
    \caption{Thirty-two-thread reduction results. Throughput uses millions of elements processed per second; the best LLP solver is selected per dataset.}
    \label{tab:reduce-llp-baseline}
\end{table}

\smallskip
\noindent\textbf{Scaling behaviour.}
Figure~\ref{fig:reduce-scaling} plots runtime and speedup on \texttt{Reduce 1M}. The sequential baseline finishes in \(1.72\)~ms; the parallel baseline variant runs in \(26.6\)~ms at 32 threads, slower due to thread-pool overhead that outweighs the benefit of parallelism on this cache-friendly workload. LLP-PTWB, by contrast, shows essentially flat performance across thread counts: it starts at \(54\)~ms with one thread and stays near \(45\)~ms through 32 threads. The LLP solver cannot exploit additional cores effectively because the reduction tree has a logarithmic critical path---once the bottom levels are drained, the remaining combines form a narrow serial chain that no amount of parallelism can accelerate, and the atomic overhead on each combine step prevents the solver from keeping up with the baseline's plain-store approach on that chain.

\begin{figure}[htb]
    \centering
    \includegraphics[width=0.48\textwidth]{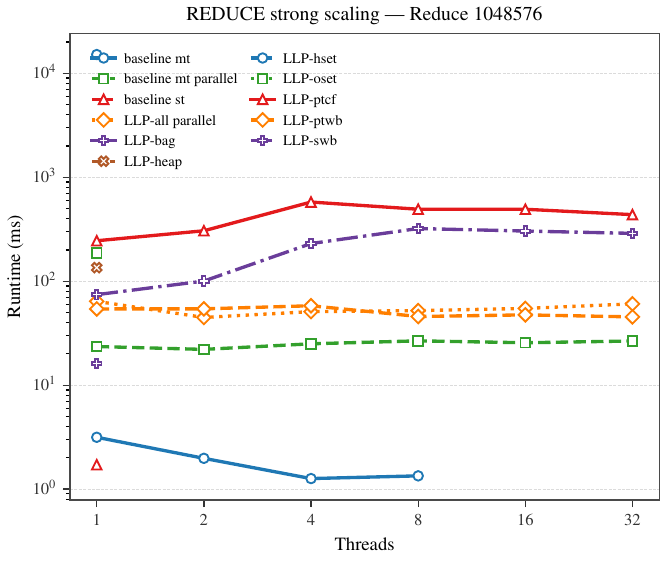}\hfill
    \includegraphics[width=0.48\textwidth]{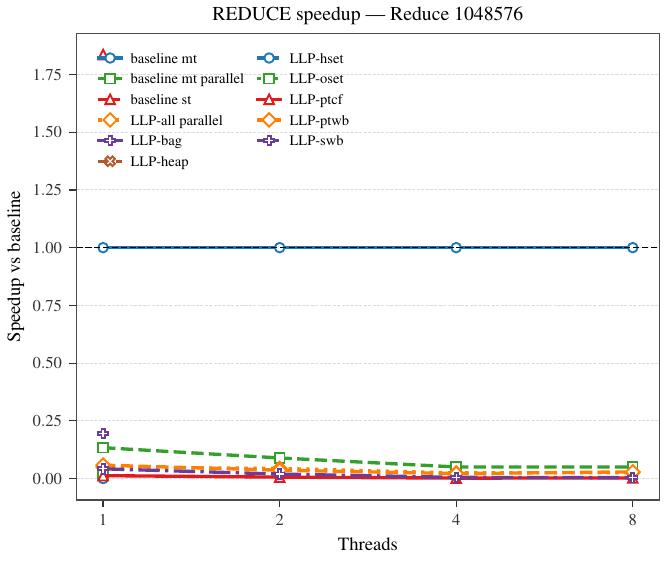}
    \caption{Strong scaling (left) and speedup over the baseline (right) for \texttt{Reduce 1M}.}
    \label{fig:reduce-scaling}
\end{figure}

\smallskip
\noindent\textbf{Worklist sensitivity.}
Figure~\ref{fig:reduce-worklists} compares worklist policies on \texttt{Reduce 1M} at 32 threads. PTWB at \(45.3\)~ms is the best LLP variant, followed by AllPar at \(60.5\)~ms, SWB at \(288.6\)~ms, and PTCF at \(436.9\)~ms. The \(6\times\) gap between PTWB and SWB shows that spatial locality still matters even in a regular workload: PTWB assigns each thread a contiguous subtree, keeping partial sums in local caches, while SWB distributes tiles across cores and incurs coherence traffic on every combine. However, even the best worklist policy cannot close the gap with the baseline, because the fundamental bottleneck is per-combine atomic overhead rather than scheduling strategy.

\begin{figure}[htb]
    \centering
    \includegraphics[width=0.48\textwidth]{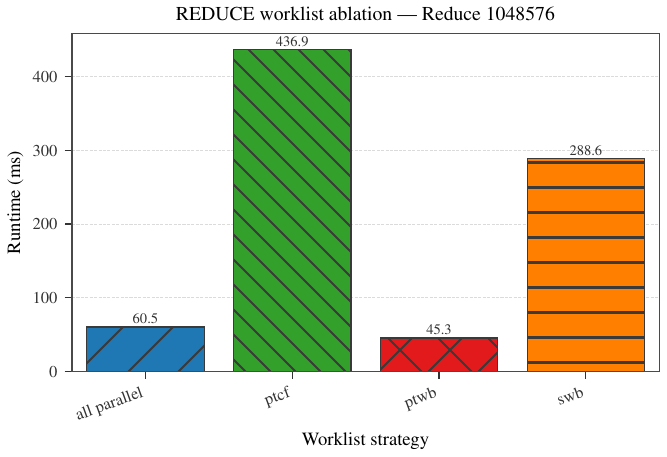}\hfill
    \includegraphics[width=0.48\textwidth]{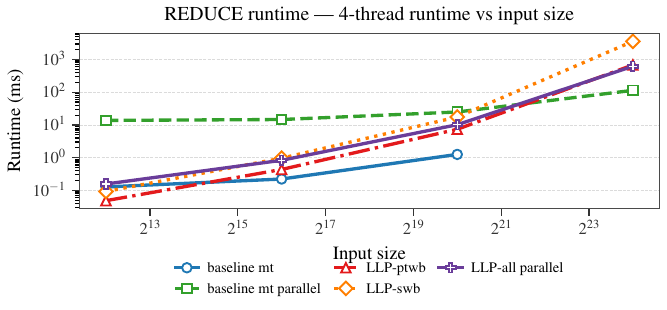}
    \caption{Worklist diagnostics on \texttt{Reduce 1M} at 32 threads (left) and runtime growth with input size at four threads (right).}
    \label{fig:reduce-worklists}
\end{figure}

\smallskip
\noindent\textbf{Why reduction is a poor fit for LLP.}
The previous algorithms all had a property that made LLP effective: an irregular frontier whose shape varied with the input, creating opportunities for demand-driven scheduling to concentrate work where it was most needed. Reduction has none of this. The reduction tree is perfectly regular, every node does the same amount of work, and the dependency structure is fixed at compile time. The LLP abstraction still produces a correct result, but its coordination machinery---atomic predicate checks, compare-and-swap combines, worklist management---adds overhead on every operation without any compensating reduction in redundant work. The baseline, by contrast, can use plain stores within each thread's subtree and synchronise only once at the root.

\smallskip
\noindent\textbf{Summary.}
Parallel reduction serves as an informative lower bound for the LLP approach. On cache-resident inputs the overhead is small enough that both approaches finish in comparable time; once the working set exceeds the LLC, the per-element atomic cost dominates and the baseline wins by \(4\)--\(5\times\). The lesson for practitioners is that LLP is not a good match for perfectly regular, bandwidth-limited computations. Its value lies in problems with irregular dependency structures where the cost of the abstraction is amortised over the savings from demand-driven scheduling---exactly the pattern observed in SSSP, BFS, Stable Marriage, and job scheduling.

\subsection{Transitive Closure}
\label{sec:closure}

\subsubsection{Formulation}
Given a directed graph \(G=(V,E)\), the transitive closure problem asks for the set of all ordered pairs \((u,v)\) such that a directed path exists from \(u\) to \(v\). The classical parallel approach is Floyd--Warshall: it iterates over intermediate vertices \(k=1,\ldots,|V|\), and at each step updates every pair \((u,v)\) by checking whether the path \(u \to k \to v\) exists. The algorithm is \(O(|V|^3)\) regardless of graph density, which means it does the same amount of work on a 500-vertex DAG with 17k edges as on a 500-vertex clique. This is wasteful on sparse graphs, where only a small fraction of the \(|V|^2\) pairs are reachable.

The LLP formulation exploits sparsity. The global state is a boolean matrix \(\C[u,v]\) indicating current reachability, ordered pointwise: once a pair is marked reachable it never reverts. A pair \((u,v)\) is \emph{forbidden} if \(\C[u,v]\) is still \texttt{false} but there exists an intermediate vertex \(w\) with \(\C[u,w] = \C[w,v] = \texttt{true}\). Advancing the state flips \(\C[u,v]\) to \texttt{true} and potentially enables further pairs higher in the lattice. On sparse DAGs, most pairs are never reachable, so the LLP solver only examines the \(O(|E| \cdot |V|)\) pairs that actually participate in reachability chains rather than all \(|V|^3\) triples. On dense graphs, however, nearly every pair eventually becomes reachable, and the advantage disappears.

\subsubsection{Implementation}
The implementation stores rows of the reachability matrix as compressed bitsets. Threads pop forbidden pairs from the worklist, splice the relevant source and destination rows with atomic \texttt{fetch\_or} operations, and push any newly discovered reachable pairs back onto the worklist. The \emph{all-parallel} solver pushes every discovered pair immediately and tends to saturate cores on sparse DAGs, while the \emph{naive, cyclic} solver uses round-robin queues that mimic a serial closure pass---cycling through source vertices and expanding each row in turn. On the smallest DAG (\texttt{DAG2}), the cyclic approach is actually competitive because the overhead of worklist management exceeds the cost of a brute-force scan; on larger instances, the demand-driven approach wins. The baseline is a parallel Floyd--Warshall that distributes the outer \(k\)-loop across threads using a shared-memory barrier at each intermediate vertex.

\subsubsection{Setup}
We benchmark three synthetic DAGs from the Galois suite and one dense social network (\texttt{Wiki-Vote} from the SNAP collection). Table~\ref{tab:closure-setup} summarises the input characteristics. All runs report the median of five executions on the platform from Section~\ref{sec:cm-eval} with thread counts from 1 to 32 (Wiki-Vote starts at 2 threads due to the memory cost of the dense reachability matrix).

\begin{table}[htb]
    \centering
    \small
    \begin{tabular}{lccc}
        \toprule
        \textbf{Dataset} & \textbf{Category} & \textbf{|V|} & \textbf{|E|} \\
        \midrule
        \texttt{DAG1} & DAG & 1\,998 & 47\,506 \\
        \texttt{DAG2} & DAG & 500 & 17\,571 \\
        \texttt{DAG3} & DAG & 1\,000 & 18\,678 \\
        \texttt{Wiki-Vote} & Social & 8\,298 & 100\,762 \\
        \bottomrule
    \end{tabular}
    \caption{Transitive closure datasets. DAG instances follow the Galois suite; \texttt{Wiki-Vote} is the SNAP social network.}
    \label{tab:closure-setup}
\end{table}

\subsubsection{Results and Discussion}
\label{sec:closure-results}

Transitive closure presents the sharpest regime split in the paper: LLP wins decisively on sparse DAGs and loses clearly on the dense social graph. The same LLP code handles both cases---only the graph density changes.

At one thread, the LLP solver achieves \(7\)--\(23\times\) speedups over Floyd--Warshall on all three DAGs. These gains are entirely algorithmic: Floyd--Warshall performs \(O(|V|^3)\) work regardless of density, while the LLP solver only processes pairs that actually become reachable through existing edges. Figure~\ref{fig:closure-scaling} reports strong scaling on \texttt{DAG1}. As thread counts grow, Floyd--Warshall's regular structure allows it to scale well (\(15.8\times\) self-speedup at 32 threads), narrowing LLP's advantage from \(23\times\) at one thread to \(2.5\times\) at 32 threads. The same pattern holds on the other DAGs: LLP's edge comes from doing less total work rather than from better parallel efficiency.

\begin{figure}[htb]
    \centering
    \includegraphics[width=0.48\textwidth]{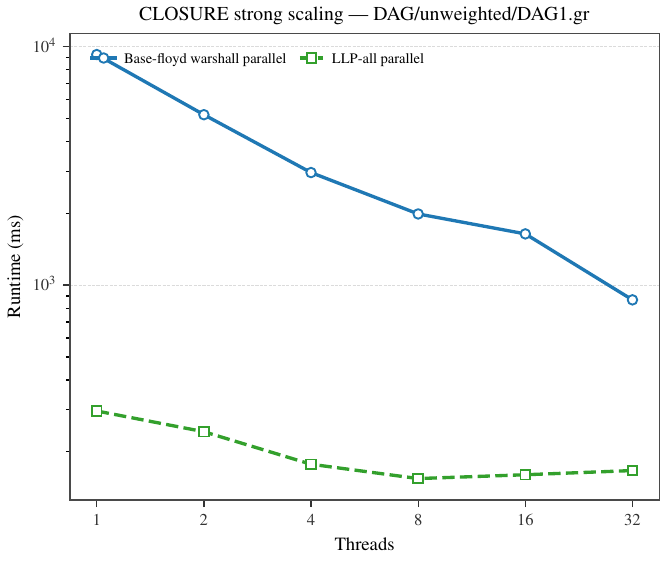}\hfill
    \includegraphics[width=0.48\textwidth]{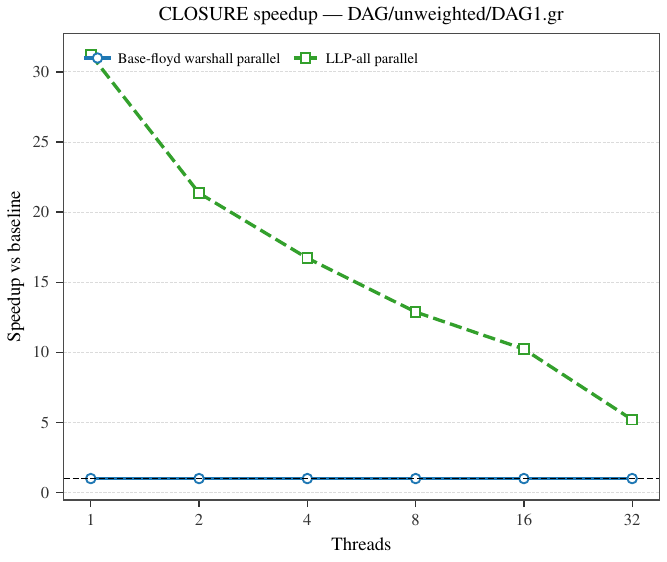}
    \caption{Runtime (left) and speedup versus Floyd--Warshall (right) for \texttt{DAG1}.}
    \label{fig:closure-scaling}
\end{figure}

\smallskip
\noindent\textbf{Dense networks: LLP loses ground.}
Figure~\ref{fig:closure-dense} shows the opposite regime on \texttt{Wiki-Vote}. The social-network community structure means that a large fraction of all \(|V|^2\) pairs are reachable, so the demand-driven approach no longer saves work but still pays the per-pair overhead of atomic operations. Floyd--Warshall wins by roughly \(0.67\times\) at 32 threads, consistent across thread counts. As with the SSSP citation-network results, when most of the problem domain is simultaneously active, the LLP overhead dominates.

\begin{figure}[htb]
    \centering
    \includegraphics[width=0.48\textwidth]{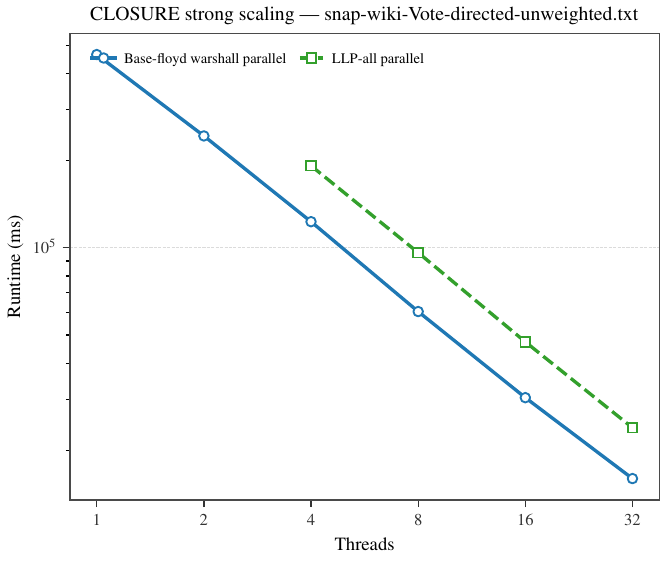}\hfill
    \includegraphics[width=0.48\textwidth]{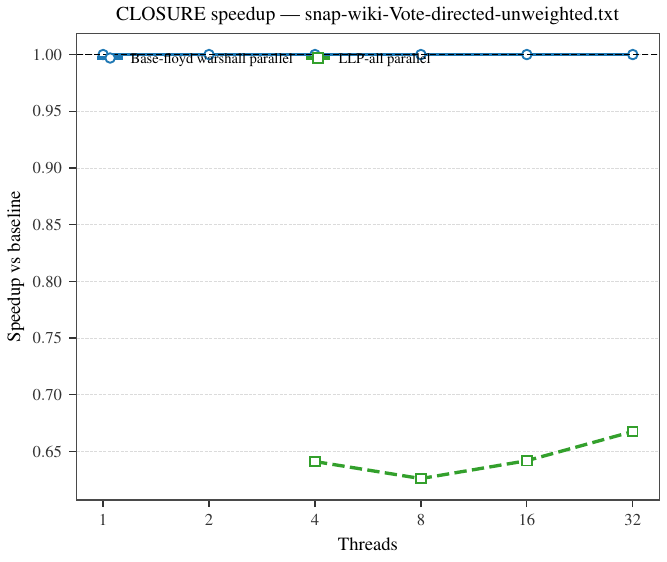}
    \caption{Runtime (left) and speedup (right) for \texttt{Wiki-Vote}. The dense edge set limits LLP scalability.}
    \label{fig:closure-dense}
\end{figure}

\smallskip
\noindent\textbf{Summary.}
Transitive closure crystallises the core trade-off of the LLP abstraction more clearly than any other benchmark in this paper. On sparse graphs, the demand-driven approach performs dramatically less work than the \(O(|V|^3)\) Floyd--Warshall baseline, delivering up to \(23\times\) single-thread gains that persist (at reduced magnitude) through 32 threads. On dense graphs with near-complete reachability, LLP's coordination overhead exceeds its work-reduction benefit and Floyd--Warshall's regular parallelism prevails. For a practitioner, the decision criterion is straightforward: if the expected reachability fraction is small relative to \(|V|^2\), LLP is the better choice; if the closure is expected to be dense, Floyd--Warshall's predictable \(O(|V|^3)\) cost is preferable.

\subsection{Knapsack}
\label{sec:knapsack}

\subsubsection{Formulation}
The 0--1 knapsack problem selects a subset of items \(i \in \{1,\ldots,n\}\) with weights \(w_i\) and values \(v_i\) that maximises total value without exceeding a capacity \(C\). The standard parallel approach fills a two-dimensional DP table row by row: for each item \(k\) and capacity \(c\), the recurrence \(\max(\C[k{-}1,c],\; v_k + \C[k{-}1,c{-}w_k])\) determines the optimal value. Rows depend only on the previous row, so all capacities within a row can be computed in parallel, but the next row cannot begin until the current one is complete. This row-synchronous structure is similar to the level-synchronous pattern in BFS and Floyd--Warshall, and suffers from the same bottleneck: threads must wait for the widest row to finish before advancing.

In the LLP interpretation, each DP cell \(\C[k,c]\) is a lattice coordinate. The lattice order is pointwise: once a cell records a higher value it never decreases. A cell becomes \emph{forbidden} when its current value is stale relative to the recurrence applied with the latest predecessor values. Advancing the state raises the cell to the tight bound and potentially enables cells in the next row. Because cells in different rows can be advanced concurrently as soon as their predecessors are ready, the LLP solver is not constrained to process one complete row before starting the next---it can overlap rows whenever the dependency chain permits.

\subsubsection{Implementation}
We tile the DP table by capacity ranges so that each tile aggregates a contiguous strip of capacities for the same item index. Every tile carries the item index \(k\), the capacity range \([c_{\min}, c_{\max})\), and the current best values stored in a packed vector of \texttt{AtomicI64}. When a tile is forbidden we atomically apply the recurrence using the latest predecessor values and publish the updated maxima. The tiling strategy determines how much independent work is available to the scheduler: narrow tiles expose more parallelism but increase worklist management overhead; wide tiles reduce overhead but limit the number of concurrent work items.

The baseline is a multi-threaded DP implementation (\texttt{baseline - dp mt}) that distributes capacity ranges across threads within each row and synchronises at row boundaries.

\subsubsection{Setup}
We evaluate six synthetic item sets with uniformly sampled weights and values. Capacities scale linearly with item count. Table~\ref{tab:knapsack-setup} summarises the workload family. Multi-thread data (1--32 threads) is available for instances with 500 or more items; smaller instances have single-thread data only. All experiments report medians over repeated Criterion samples on the platform described in Section~\ref{sec:cm-eval}.

\begin{table}[htb]
    \centering
    \small
    \begin{tabular}{lccc}
        \toprule
        \textbf{Dataset} & \textbf{Items} & \textbf{Capacity} & \textbf{Fill ratio} \\
        \midrule
        \texttt{n100\_cap250} & 100 & 250 & 3.54 \\
        \texttt{n500\_cap800} & 500 & 800 & 5.58 \\
        \texttt{n1000\_cap1500} & 1\,000 & 1\,500 & 5.99 \\
        \texttt{n1000\_cap1600} & 1\,000 & 1\,600 & 5.61 \\
        \texttt{n2000\_cap3200} & 2\,000 & 3\,200 & 5.61 \\
        \texttt{n4000\_cap6400} & 4\,000 & 6\,400 & 5.62 \\
        \bottomrule
    \end{tabular}
    \caption{Knapsack benchmarks derived from uniformly sampled weights and values. Fill ratio denotes the mean weight-to-capacity ratio across items.}
    \label{tab:knapsack-setup}
\end{table}

\subsubsection{Results and Discussion}
\label{sec:knapsack-results}

Knapsack presents a different performance dynamic from the previous algorithms. The LLP solver carries a significant single-thread overhead, but the baseline's own scaling behaviour is so poor that LLP overtakes it at high thread counts. The result is a win for LLP, but one that depends more on the baseline's weakness than on LLP's strength.

\smallskip
\noindent\textbf{Single-thread overhead.}
At one thread, the DP baseline is substantially faster than LLP on every instance. On \texttt{n500\_cap800}, the baseline finishes in \(0.57\)~ms versus \(2.69\)~ms for the best LLP solver (PTWB), a \(4.7\times\) overhead. On the largest instance (\texttt{n4000\_cap6400}), the gap is similar: \(34.2\)~ms versus \(169.4\)~ms (\(5.0\times\)). The overhead comes from two sources. First, every DP-cell update is performed through an atomic compare-and-swap on a packed \texttt{AtomicI64}, whereas the baseline uses plain stores within each thread's capacity range. Second, the LLP solver evaluates the forbiddenness predicate for each tile it pops---reading predecessor values from the previous row---before deciding whether an update is needed. The baseline avoids this check because the row-synchronous structure guarantees that predecessors are always ready.

\smallskip
\noindent\textbf{Baseline comparison at 32 threads.}
Table~\ref{tab:knapsack-llp-baseline} reports the best LLP configuration versus the DP baseline at 32 threads. LLP wins on all four multi-thread instances, with speedups ranging from \(1.82\times\) to \(3.79\times\). The wins are real but their interpretation requires understanding why the baseline degrades.

\begin{table}[htb]
    \centering
    \resizebox{\linewidth}{!}{%
    \begin{tabular}{lcccccccc}
        \toprule
        \textbf{Dataset} & \textbf{Items} & \textbf{Capacity} & \textbf{Baseline (ms)} & \textbf{Baseline (M items/s)} & \textbf{Best LLP} & \textbf{LLP (ms)} & \textbf{LLP (M items/s)} & \textbf{Speedup} \\
        \midrule
        \texttt{n500\_cap800} & 500 & 800 & 37.45 & 0.01 & SWB & 9.88 & 0.05 & 3.79\,\(\times\) \\
        \texttt{n1000\_cap1600} & 1\,000 & 1\,600 & 87.60 & 0.01 & SWB & 32.11 & 0.03 & 2.73\,\(\times\) \\
        \texttt{n2000\_cap3200} & 2\,000 & 3\,200 & 199.48 & 0.01 & Buckets & 109.44 & 0.02 & 1.82\,\(\times\) \\
        \texttt{n4000\_cap6400} & 4\,000 & 6\,400 & 564.77 & 0.01 & Buckets & 232.65 & 0.02 & 2.43\,\(\times\) \\
        \bottomrule
    \end{tabular}}
    \caption{Thirty-two-thread comparison between the dynamic-programming baseline and the fastest LLP solver per dataset. Throughput is expressed in millions of items processed per second.}
    \label{tab:knapsack-llp-baseline}
\end{table}

\smallskip
\noindent\textbf{Scaling behaviour: LLP wins because the baseline anti-scales.}
Figure~\ref{fig:knapsack-scaling} reports strong scaling on \texttt{n4000\_cap6400}. The most striking feature is the baseline's trajectory: it starts at \(34.2\)~ms with one thread, improves to \(19.5\)~ms at two threads, then \emph{degrades} steadily---\(23.4\)~ms at four, \(62.2\)~ms at eight, \(245.4\)~ms at sixteen, and \(564.8\)~ms at thirty-two threads. This is severe anti-scaling: the 32-thread baseline is \(16.5\times\) slower than its own single-thread performance. The cause is the row-synchronous barrier. Each row of the DP table contains only \(C = 6{,}400\) capacity cells, distributed among 32 threads that is only 200 cells per thread. The work per row per thread is so small that the cost of synchronising at the row boundary dominates actual computation time, and this overhead grows with thread count.

The LLP solver avoids this trap. Its runtime is roughly flat from one through eight threads (\(169\)~ms), reflecting the overhead of atomic operations on what is still a small working set. At 32 threads, the bucketed scheduler reaches \(232.7\)~ms. The LLP solver does not \emph{improve} much with threads either---the DP table does not expose enough independent work for 32 cores---but it crucially does not degrade. The speedup at 32 threads (\(2.43\times\)) therefore comes primarily from the baseline getting worse rather than from LLP getting better.

\begin{figure}[htb]
    \centering
    \includegraphics[width=0.48\textwidth]{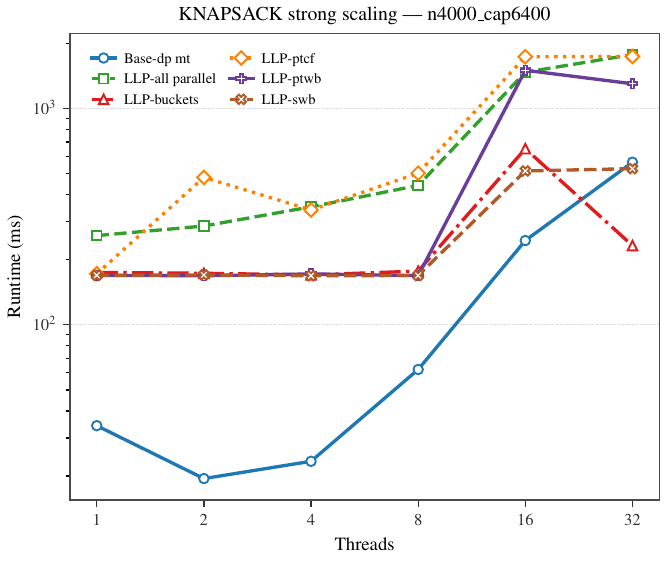}\hfill
    \includegraphics[width=0.48\textwidth]{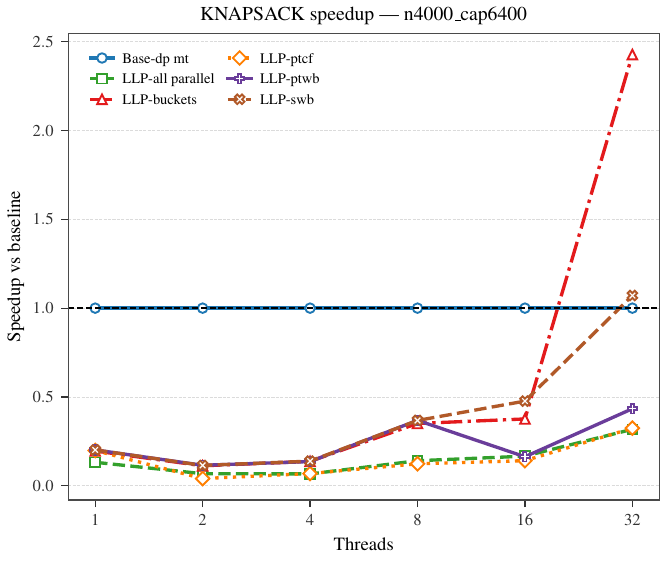}
    \caption{Strong scaling (left) and speedup relative to the dynamic-programming baseline (right) on \texttt{n4000\_cap6400}.}
    \label{fig:knapsack-scaling}
\end{figure}

On the smaller \texttt{n500\_cap800}, the same pattern appears in compressed form: the baseline degrades from \(0.57\)~ms at one thread to \(37.5\)~ms at 32 threads (\(66\times\) anti-scaling), while LLP stays near \(2.7\)--\(9.9\)~ms. The LLP crossover happens at eight threads (\(1.51\times\)), reaching \(3.79\times\) at 32 threads.

\smallskip
\noindent\textbf{Worklist sensitivity.}
Figure~\ref{fig:knapsack-worklists} compares worklist policies on \texttt{n4000\_cap6400} at 32 threads. Bucketed scheduling is the clear winner at \(232.7\)~ms, followed by SWB at \(527.4\)~ms, PTWB at \(1{,}301.8\)~ms, and PTCF at \(1{,}740.1\)~ms. The ranking is different from all previous algorithms: buckets win because they group tiles by capacity range, so threads processing adjacent capacity strips benefit from spatial locality in the DP table. SWB is the runner-up because its global queue provides reasonable load balance on the narrow tile frontier. PTWB, which dominated in SSSP and BFS, performs poorly here because its per-thread partitioning scatters non-adjacent tiles across cores---there is no spatial frontier to exploit, only a dependency front that advances one item row at a time.

On smaller instances the picture reverses: on \texttt{n500\_cap800} at 32 threads, SWB wins (\(9.88\)~ms) because the capacity range is narrow enough that all tiles fit in shared cache regardless of ordering. This sensitivity to instance size is characteristic of DP workloads where the available parallelism is bounded by the table width.

\begin{figure}[htb]
    \centering
    \includegraphics[width=0.48\textwidth]{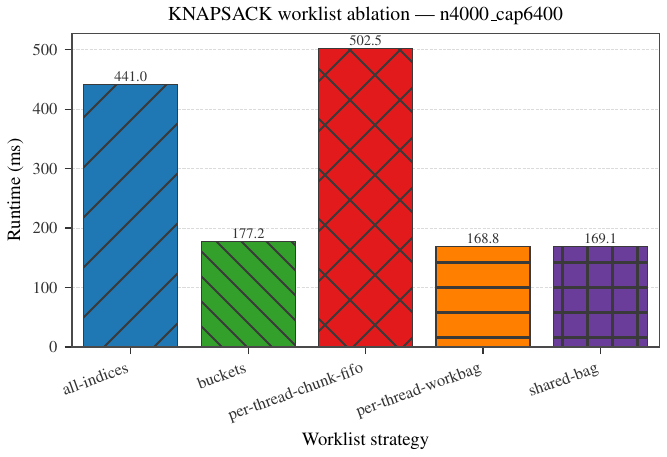}\hfill
    \includegraphics[width=0.48\textwidth]{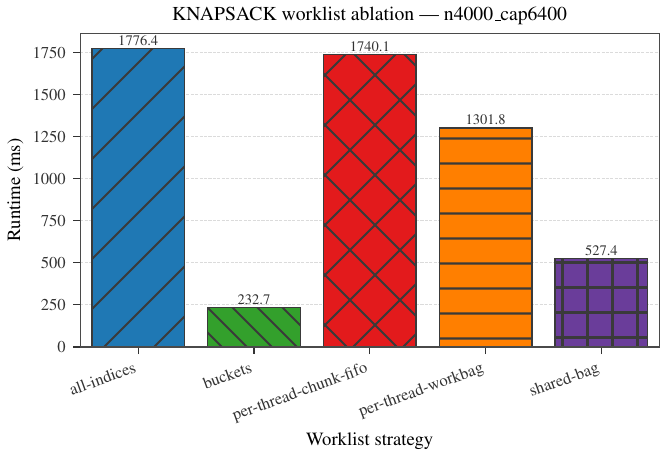}
    \caption{Worklist diagnostics on \texttt{n4000\_cap6400}. Left: eight threads; right: 32 threads. Bucketed queues keep more tiles in flight once thread count grows.}
    \label{fig:knapsack-worklists}
\end{figure}

\smallskip
\noindent\textbf{Summary.}
Knapsack is a qualified success for LLP. The framework wins at 32 threads on every large instance (\(1.82\)--\(3.79\times\)), but the win comes from the baseline's severe anti-scaling rather than from LLP's own parallel efficiency. At low thread counts, the DP baseline is faster and LLP adds measurable overhead. The practical lesson is that LLP is a good choice for DP workloads that will be deployed at moderate-to-high thread counts, especially when the table structure leads to synchronisation-heavy baselines. The worklist recommendation also differs from the graph algorithms: bucketed scheduling, which groups tiles by problem structure, outperforms the per-thread partitioning strategies that dominated in SSSP and BFS.

\section{Cross-Algorithm Analysis}
\label{sec:cross-analysis}

Having evaluated the seven algorithms in the preceding sections, we now look at the results as a whole and try to identify the patterns that determine when LLP-FW does well and when it does not. The algorithms span graph traversal (SSSP, BFS), combinatorial matching (Stable Marriage), DAG scheduling (Job Scheduling), tree-structured computation (Parallel Reduction), set closure (Transitive Closure), and dynamic programming (Knapsack). Each of these algorithms exercises the LLP abstraction differently and the individual results already hint at some trends.

\subsection{Where LLP excels}

Looking across the results, we observe that LLP-FW performs best when the following conditions are present.

The clearest wins are on problems where only a small fraction of the domain is forbidden at any given time, i.e. when the forbidden frontier is narrow or structured. On road-network SSSP and BFS, the frontier is long but thin---each level contains relatively few forbidden vertices spread over a large graph diameter. The LLP solver exploits this by giving each thread a private queue partition, which eliminates the global-queue contention that hurts the baseline at high thread counts. We can see this in the BFS results: \(16\times\) on road-network BFS at 32 threads. SSSP on the same road network yields a more modest \(1.5\times\) at 32 threads, limited by the higher per-vertex cost of weighted relaxation. Sparse-DAG transitive closure shows the same effect from a different angle: the reachability set is small relative to \(|V|^2\), so LLP processes far fewer pairs than Floyd--Warshall's fixed \(O(|V|^3)\) sweep, yielding up to \(23\times\) at one thread.

Another favorable scenario is when there are cascading dependencies that serialise round-based baselines. Stable Marriage is the extreme case. The parallel Gale--Shapley baseline has to wait for a round of proposals to complete before processing the displacements they cause, but in LLP a displaced man is re-enqueued immediately and can re-propose in the same sweep. This overlap of cascading proposals is what delivers \(246\times\) at 32 threads on the 1\,000-participant instance, and \(109\times\) on the 10\,000-participant instance. Job scheduling shows a milder version of the same effect: the level-synchronous baseline is effectively flat across thread counts on \texttt{Job 10000}, while LLP's demand-driven propagation resolves dependency chains without global synchronisation, producing \(17.8\times\) even at one thread.

Memoization through fixed states also plays an important role. When the problem structure guarantees that an index, once settled, will never become forbidden again, the fixed-state bit vector eliminates redundant predicate checks entirely. Job scheduling benefits the most: each job commits exactly once, so the solver's effective complexity scales with the number of edges in the DAG rather than with repeated vertex scans. This is why LLP outperforms the topological baseline \emph{before any parallelism is applied}---the algorithmic savings dominate the atomic overhead.

\subsection{Where LLP is weak}

On the other hand, we observe that the LLP advantage erodes or disappears in two scenarios.

The first is when the forbidden frontier is wide and persistent on dense inputs. When most of the problem domain is simultaneously forbidden, the LLP solver pays atomic overhead on every operation without a compensating reduction in total work. We can see this in citation-network SSSP (\texttt{coAuthorsDBLP}: \(0.50\times\) at 32 threads), dense-social-network transitive closure (\texttt{Wiki-Vote}: \(0.67\times\)), and large Kronecker BFS (\texttt{kron-2e22}: initially trailing until 16 threads). In each of these cases the baseline's regular, synchronised sweeps are more efficient since they avoid the per-element CAS cost and benefit from predictable memory-access patterns.

The second scenario is when the workload is regular and bandwidth-limited, with no frontier to exploit. Parallel reduction is the purest example of this. The reduction tree is perfectly regular, every node does the same amount of work, and the dependency structure is fixed at compile time. LLP's coordination machinery adds overhead on every combine step without any opportunity for demand-driven scheduling to save work. The baseline wins by \(4\)--\(5\times\) on the largest reduction. Note that this is not a failure of the LLP \emph{formulation}---the solver does produce correct results---but it shows that the abstraction's value lies in irregular problems where the cost of generality is amortised.

\subsection{Scheduler selection across algorithms}

Another recurring observation is that no single worklist policy dominates all problems. Table~\ref{tab:cross-scheduler} summarises the best scheduler per algorithm and the mechanism behind it.

\begin{table}[htb]
    \centering
    \resizebox{\linewidth}{!}{%
    \begin{tabular}{lll}
        \toprule
        \textbf{Algorithm} & \textbf{Best scheduler} & \textbf{Mechanism} \\
        \midrule
        SSSP & PTWB & Recency bias: prioritise vertices whose distances just decreased \\
        BFS & PTWB & Spatial partitioning: contiguous neighbour slabs stay cache-local \\
        Stable Marriage & PTCF & Chunked FIFO: amortise cascade re-evaluation overhead \\
        Job Scheduling & Any (within 15\%) & Memoization eliminates redundant work; scheduling matters less \\
        Parallel Reduction & PTWB & Subtree locality, though no policy closes the baseline gap \\
        Transitive Closure & All-parallel & Maximise pair discovery rate on sparse frontiers \\
        Knapsack & Buckets & Group tiles by capacity for DP-table spatial locality \\
        \bottomrule
    \end{tabular}}
    \caption{Best-performing LLP scheduler per algorithm. The dominant mechanism differs in each case, reflecting the diversity of frontier structures across problem domains.}
    \label{tab:cross-scheduler}
\end{table}

As we can see from the table, the optimal scheduler tracks the \emph{structure of the forbidden frontier} rather than the problem's formal domain. Graph algorithms with narrow frontiers do best with per-thread partitioning (PTWB); matching problems with long cascades do best with chunked FIFOs (PTCF); DP problems with regular tile grids work best with bucketed grouping; and algorithms where memoization eliminates re-processing are insensitive to the scheduler choice. This suggests that an adaptive scheduler that detects frontier shape at runtime could capture most of the per-algorithm tuning benefit automatically, which is a direction we leave for future work.

\subsection{The cost of generality}

Across all seven algorithms, we observe that LLP-FW incurs a per-operation overhead (atomic reads, compare-and-swap updates, worklist management) that specialised baselines avoid. In the best cases (job scheduling, sparse closure, road-network BFS), this overhead is dwarfed by algorithmic savings or baseline anti-scaling. In the worst case (parallel reduction), the overhead is the dominant cost and LLP loses. In the middle ground (power-law SSSP, citation BFS, knapsack), the outcome depends on thread count and input density.

This trade-off is inherent to the framework's design: the same lock-free predicate-check-and-advance loop handles all seven algorithms without modification. A bespoke implementation for any single algorithm could avoid atomic operations where contention is provably absent, use problem-specific data structures, and eliminate the worklist abstraction entirely. The value of LLP-FW is that it \emph{amortises engineering effort} across problems: implementing a new lattice-linear algorithm means inheriting the parallel runtime, the worklist library, and the memoization infrastructure without writing any concurrency code. Our evaluations show that this generality is free or beneficial on five of seven workloads at moderate-to-high thread counts, and costly only on the most regular, bandwidth-limited computation in the suite.

\section{Conclusion}
\label{sec:cm-conclusion}

In this paper we introduced LLP-FW, a lock-free shared-memory framework that separates problem-specific logic (forbiddenness predicates and advance rules) from the parallel solver policy (worklists, scheduling, and concurrency strategy). Using a single Rust runtime built on atomic compare-and-swap operations, we instantiated and evaluated seven algorithms---Single-Source Shortest Path, Breadth-First Search, Stable Marriage, Job Scheduling, Parallel Reduction, Transitive Closure, and 0--1 Knapsack---without rewriting any synchronisation code between problems.

\subsection{Summary of results}

The evaluation shows that the results fall into roughly three categories depending on the problem characteristics.

For problems with narrow or structured frontiers and cascading dependencies, LLP-FW delivers order-of-magnitude improvements over the baselines. Stable Marriage reaches \(246\times\) at 32 threads on the 1\,000-participant instance and \(109\times\) on the 10\,000-participant instance with near-linear self-speedup (\(27.9\times\) across 32 threads), because the LLP solver overlaps cascading re-proposals that the round-based Gale--Shapley baseline must serialise. Sparse transitive closure achieves up to \(23\times\) at one thread on DAG workloads because the demand-driven solver examines only reachable pairs rather than the \(O(|V|^3)\) triples that Floyd--Warshall processes. Road-network BFS reaches \(16\times\) at 32 threads thanks to the elimination of global-queue contention on narrow frontiers, and SSSP on sparse power-law graphs reaches up to \(4.7\times\). Job scheduling delivers \(17.8\times\) at one thread on deep DAGs through memoization that prevents redundant re-evaluation of settled jobs.

For problems where the baseline scales poorly or where the gains are more moderate, LLP-FW still wins but by smaller margins. Knapsack achieves \(1.8\)--\(3.8\times\) at 32 threads, but the gains largely reflect the DP baseline's severe anti-scaling under row-synchronous barriers. Citation-network BFS delivers \(2.7\times\) at 32 threads, with the crossover occurring around eight threads as LLP's overhead is amortised. Power-law SSSP and Kronecker BFS show modest gains (\(1.2\)--\(1.6\times\)) that follow the density-dependent transition between narrow and wide frontiers.

For regular, bandwidth-limited workloads, LLP-FW loses to specialised baselines. Parallel reduction---the most regular workload in the suite---runs \(4\)--\(5\times\) slower than the cache-optimised baseline at 32 threads since the per-combine atomic overhead is paid on every operation with no compensating reduction in work. Citation-network SSSP (\(0.5\times\)) and dense-graph transitive closure (\(0.67\times\)) also lose because wide, persistent frontiers cause the atomic coordination cost to dominate.

\subsection{Key insights}

Looking across all the results, a few observations stand out.

Perhaps the most important one is that \textbf{frontier shape predicts LLP effectiveness} more reliably than problem domain or input size. Narrow, structured frontiers (road networks, sparse DAGs) consistently favour LLP while wide, dense frontiers (citation networks, social graphs) consistently favour synchronous baselines. This is useful for practitioners considering the framework for a new problem: if the workload has narrow frontiers where only a small fraction of elements are forbidden at any time, LLP-FW is likely a good fit.

We also observe that \textbf{scheduler choice matters a lot} and varies across algorithms. PTWB dominates on graph traversal through recency bias, PTCF wins on matching through cascade amortisation, buckets win on knapsack through spatial grouping, and job scheduling is insensitive to the scheduler because memoization eliminates the work that scheduling would otherwise optimise. No single policy is universally best, which means the scheduler is not an implementation detail but a design choice that has to be matched to the problem.

Finally, we observe that \textbf{the cost of generality is bounded}. LLP-FW's per-operation overhead (atomic reads, compare-and-swap, worklist management) is the price of a single runtime serving all seven algorithms. On five of seven workloads at moderate-to-high thread counts, this cost is offset by algorithmic benefits or baseline weaknesses. On the remaining two (reduction and dense-frontier cases), the cost dominates. Importantly though, the overhead does not grow with problem size---it is a constant factor per operation---so it becomes proportionally smaller on larger, more irregular instances.

\subsection{Future work}

There are several directions that follow naturally from our results. As discussed in Section~\ref{sec:cross-analysis}, frontier shape determines the optimal scheduler, yet the current implementation requires the user to choose a policy before the solver starts. An interesting extension would be a runtime monitor that samples frontier width and depth periodically and switches policies mid-computation. The main challenge here is making the switch lock-free: draining one worklist into another without losing or duplicating items requires careful coordination, but the existing atomic infrastructure provides the building blocks.

Another direction is reducing the atomic overhead on low-contention indices. The single-thread overhead we observe in SSSP (\(\approx\)10\(\times\)) and knapsack (\(\approx\)5\(\times\)) comes largely from compare-and-swap operations on indices that no other thread is contending for. A speculative optimisation could allow the solver to use plain stores when a thread can prove it holds exclusive access to a region of the state vector, falling back to CAS only when contention is detected.

The weak results on citation networks (SSSP, BFS) and dense social graphs (transitive closure) suggest that LLP-FW would benefit from a frontier-width-aware strategy for these dense-frontier cases. One possibility is a hybrid solver that uses LLP's demand-driven approach while the frontier is narrow and falls back to a synchronous bulk-processing mode when frontier width exceeds a threshold.

Beyond shared-memory multicore machines, the LLP abstraction---local predicates, monotone advances, demand-driven scheduling---maps naturally onto message-passing systems where each node maintains a partition of the global state and exchanges updates with neighbours. GPU acceleration is also a candidate since the predicate check and advance step are both data-parallel operations that could benefit from SIMT execution, provided the worklist can be managed efficiently on the device.

Finally, the seven algorithms in this paper were chosen to span a range of problem structures, but they do not exhaust the space of lattice-linear problems. Network flow, constraint satisfaction, and certain game-theoretic equilibrium computations have been shown to admit LLP formulations in the theoretical literature. Implementing these within LLP-FW and evaluating them against specialised baselines would further clarify the boundaries of the framework's applicability.

% Keep SSSP exploratory figures within SSSP section; hide global copies

% Footer for self-contained arxiv version (skipped in thesis mode)
%
%%%%%%%%%%%%%%%%%%%%%%%%%%%%%%%%%%%%%%%%%%%%%%%%%%%%%%%%%%%%%%%%%%%%%%
% Generate the bibliography.					     %
%%%%%%%%%%%%%%%%%%%%%%%%%%%%%%%%%%%%%%%%%%%%%%%%%%%%%%%%%%%%%%%%%%%%%%
%								     %
%\nocite{*}
\bibliographystyle{plain}
\bibliography{diss}

\clearpage
\appendix

\end{document}